%% file: CoopLocNLOSJrRevised_arxiv.tex
\documentclass[journal,onecolumn, 11pt]{IEEEtran}
\input{packages_all.tex}

\normalsize
\usepackage{floatflt}
\usepackage{wrapfig}
\usepackage{algorithm}
\usepackage{algorithmic}
\usepackage{stfloats}

\begin{document}

\title{Collaborative High Accuracy Localization in Mobile Multipath Environments}

\author{ \IEEEauthorblockN{Venkatesan. N. Ekambaram*, Kannan Ramchandran* and Raja Sengupta**}\\
\IEEEauthorblockA{*Department of EECS, University of California, Berkeley\\
**Department of CEE, University of California, Berkeley\\
Email: \{venkyne, kannanr\}@eecs.berkeley.edu}, sengupta@ce.berkeley.edu}

% make the title area
\maketitle

\begin{abstract}

We study the problem of high accuracy localization of mobile nodes in a multipath-rich environment where sub-meter accuracies are required. We employ a peer-to-peer framework where the vehicles/nodes can get pairwise multipath-degraded ranging estimates in local neighborhoods. The challenge is to overcome the multipath-barrier using redundancy in measurements across time and space enabled through cooperation, in order to provide the desired accuracies especially under severe multipath conditions when the fraction of received signals corrupted by multipath is dominating.  We invoke an analytical {\em graphical} model framework based on \emph{particle filtering} and validate its potential for high accuracy localization through simulations. We also address design questions such as ``How many anchors and what fraction of line-of-sight (LOS) measurements are needed to achieve a specified target accuracy?", by showing that the Cramer-Rao Lower Bound for localization can be expressed as a product of two factors - a scalar function that depends only on the parameters of the {\em noise} distribution and a matrix that depends only on the {\em geometry} of node locations and the underlying {\em connectivity graph}.  A simplified expression is obtained for the CRLB,  which provides an insightful understanding of the bound and helps deduce the scaling behavior  of the estimation error as a function of the number of agents and anchors in the network.\\

 {\bf {\em Keywords}} -  Localization, Multipath, Graphical Models, Particle Filtering, Cramer Rao Bound.
\end{abstract}

\IEEEpeerreviewmaketitle

\section{Introduction}
High-accuracy localization is mandated in many applications like vehicle safety  \cite{wilson1998potential},  autonomous robotic systems \cite{nerurkar2009distributed},   Unmanned Air Vehicle (UAV) systems etc, where sub-meter accuracies are needed. Standard GPS receivers can have localization errors of over fifty or more meters which is unacceptable for many of these applications. The principal problem is multipath interference \cite{chen1999non}, which is particularly prevalent in cities and ``urban canyon'' environments, and corrupts a large fraction of the measurements. Many of the existing solutions such as D-GPS, A-GPS, N-RTK \cite{rizos2002network} etc, that augment the GPS system are typically expensive and further fail to address multipath. Our goal here is to design algorithms that can provide sub-meter accuracies in severe multipath environments with minimal communication overhead given the bandwidth constraint in applications such as vehicular safety enabled by technology like DSRC (Dedicated Short Range Communication) \cite{kenney2011dedicated}. Our design philosophy is a ``peer to peer'' architecture (see Fig \ref{fig:vanet}) where nodes collaborate and help each other refine their position estimates. Collaboration coupled with mobility generates a large pool of measurements in the system. The fundamental insight is that, some fraction of these measurements will be produced by line-of-sight (LOS) dominated signals, and hence be fairly accurate, while some fraction will be corrupted by dominated non-line-of-sight (NLOS) reflected waves.  Receivers do not know a priori which measurements are LOS and which are NLOS. Hence, the task of the users is to cooperatively discard the NLOS signals, thus enabling them to compute high-precision position estimates. Our first main contribution in this paper is to uncover a framework and a distributed algorithm founded on \emph{graphical models} for {\em collaborative narrowband NLOS} localization with {\em minimal messaging} overhead. The proposed algorithm shows promise of sub-meter accuracies in a realistic simulation setup.\\

\begin{figure}
\centering
\includegraphics[height=1.5in]{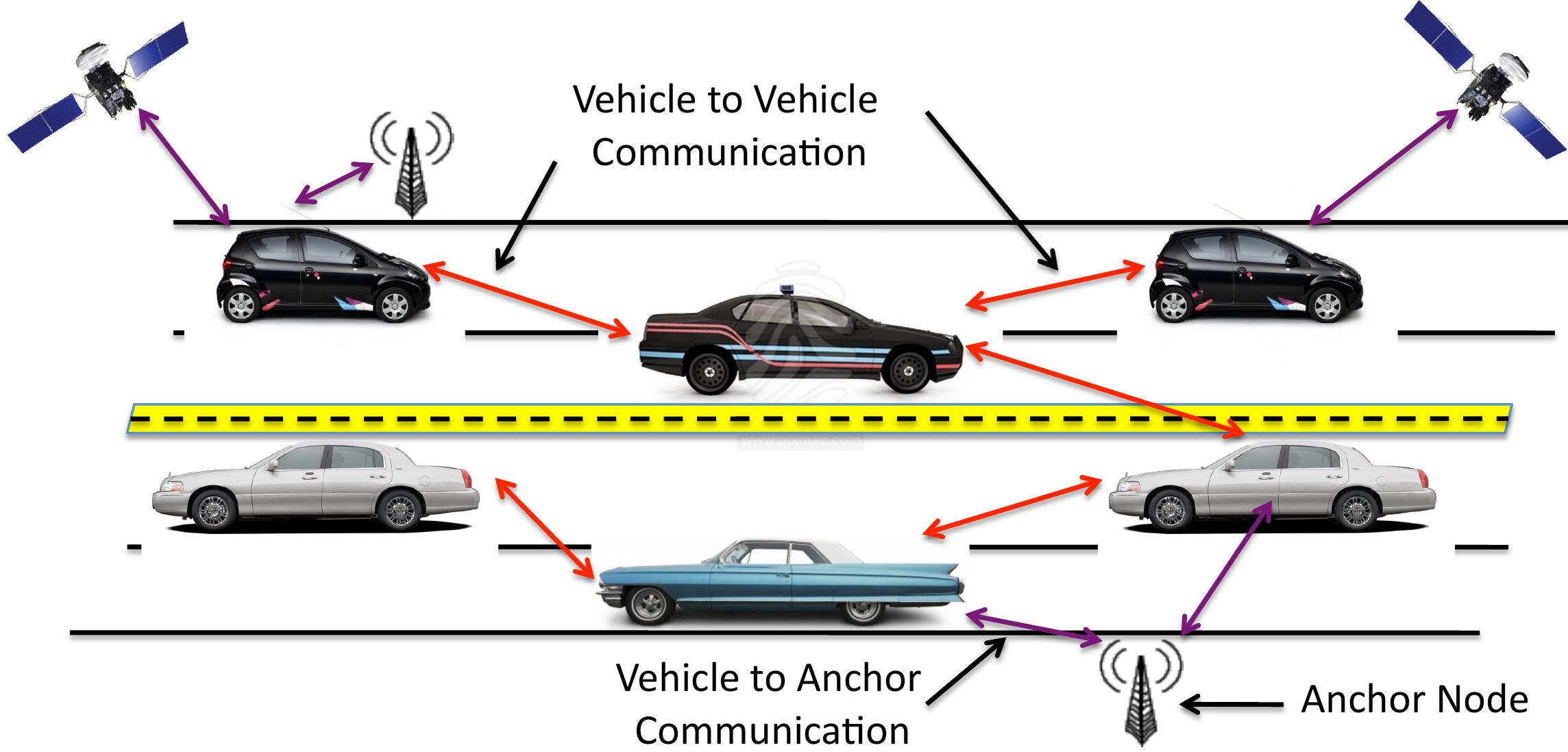}
\caption{ Peer-to-Peer Collaborative Localization}
\label{fig:vanet}
\vspace*{0in} \end{figure}

Our second main contribution in this paper is a theoretical characterization of the localization accuracy as a function of the number of {\em anchors} (nodes with known locations), agents/vehicles and the fraction of LOS measurements by analyzing the Cramer-Rao Lower Bound (CRLB) \footnote{The CRLB is mostly tight in the high signal-to-noise (SNR) ratio regime and is not a good indicator of the estimator performance at low SNR. Other bounds like the Ziv-Zakai and Barankin bounds give a good characterization at low SNR but are quite complex to offer any insights.}, a fundamental lower bound on the best possible mean square error achievable using an unbiased estimator. For a generalized distance/angle measurement model,  we provide a simplified characterization of the scaling behavior of the bound as a function of the number of nodes in the network and the fraction of LOS measurements. Further we show a ``separation principle''  under simplifying assumptions, wherein the effect of the {\em node geometry} and the {\em noise distribution} can be independently analyzed. \\

The paper is organized as follows. Section II discusses the related work in this area and the positioning of our work in this literature. Section III sets up the problem formulation and notations. Section IV describes our inference algorithm for NLOS localization. Section V summarizes our theoretical results and bounds on the localization accuracy. Section VI provides simulation results comparing the performance of the algorithm with the derived theoretical bounds. 

\section{Related work}
There has been significant research work focusing on cooperative localization algorithms particularly for the LOS case \cite{patwari2005locating, parker2007cooperative, alam2011dsrc}.  However NLOS localization has received relatively lesser focus in the literature. There has been quite some work on non-cooperative (i.e. no collaboration between mobile agents) NLOS localization \cite{chen1999non, ananthasubramaniam2008cooperative, casas2006robust} . A major part of this literature focuses on consensus based RANSAC (RAndom SAmple Consensus) like algorithms \cite{casas2006robust} which in effect assigns a 0-1 weight to the measurements detecting them to be LOS or NLOS. The complexity of this combinatorial approach explodes with increasing fraction of NLOS measurements especially in the mobile setting. Our algorithm can be conceptually thought of assigning a ``soft-weight'' to the measurements \cite{ismail2007nlos, khodjaev2010survey}  which helps get the complexity under control and further exploit the statistics of the measurements.\\

There is also some work on using ultra-wide-band (UWB) \cite{wymeersch2009cooperative} or multi-antenna array systems \cite{seow2008non, quek, ding2012} for localization where different arriving paths can be distinguished. Wymeersch et al. \cite{wymeersch2009cooperative} and consequent follow-up works \cite{wymeersch2012} propose a graphical model framework for the problem of interest which though has some similarities to our work, is quite different in the following aspects. We focus on narrowband environments leading to a different noise model that is relevant in vehicular safety applications. We also model correlation of readings across time that arises in the vehicular setup \cite{abbas2012measurement} and our algorithm efficiently handles this while existing work does not consider correlation. Further, our concern is also towards minimizing the communication overhead and our algorithm only requires a minor overhead to the existing DSRC message that needs to be transmitted periodically every $100ms$ as prescribed by the standard. There is more recent work focusing on convex optimization frameworks for the problem of interest \cite{venkymc} which are analytically tractable but computationally hard and communication intensive especially in mobile environments.\\

 There is considerable theoretical work addressing fundamental performance bounds for LOS localization \cite{savvides2003error, patwari2005locating, botteron2004cramer, chang2004estimation, weiss2008improvement}.  However, the complex nature of the expressions renders it difficult to gain insights on its scaling behavior as a function of the number of nodes. Further, there is relatively less literature for the NLOS setting \cite{qi2006analysis, shen2009fundamental2}. Shen et al. \cite{shen2009fundamental2} analyze the CRLB for UWB cooperative localization. Qi et al. \cite{qi2006analysis} consider a narrowband non-cooperative case with assumption on the exact knowledge of which measurements are NLOS that is difficult to get in practice especially in narrowband environments which is the focus of this work.  In practice, it might be reasonable to assume that one has prior knowledge of the statistics of the NLOS distribution. For example, exponential noise models for NLOS have been proposed in the literature \cite{saleh1987statistical}. 
An analytical characterization of the performance as a function of the NLOS noise distribution and the fraction of LOS measurements is missing in the literature, which we address here. One of our contributions is to provide concrete insights into the scaling behavior of the bound as a function of the number of anchors, vehicles and the fraction of LOS measurements. An arxiv \cite{ekambaram2012cooploc} version of this paper contains detailed proofs and descriptions. Parts of this work have appeared in conference proceedings \cite{ekambaram2010distributed}, \cite{DBLP:conf/globecom/EkambaramRS11}.

\section{Problem Setup}
\label{sec:ProbSet}
We consider a network of $N$ mobile agents and $M$ static anchors with known locations (see Fig. \ref{fig:vanet}). Let $u_t(k) \in {\mathbb{C}}$ denote the true location of the $k$th vehicle at time instant $t$, where the real part represents the $x$-coordinate and the imaginary part the $y$-coordinate. Measurements of the form  $\theta_t(km) = f(u_t(k), u_t(m)) + n_t(km),$ are obtained between vehicles that are within a communication radius $R$. The function $f(.)$ depends on the measurement modality. For example, $f(u_t(k),u_t(m)) = ||u_t(k)-  u_t(m)||$ in the case of Time of Arrival based sensors, $f(u_t(k),u_t(m)) \propto \frac{1}{||u_t(k)-u_t(m)||^\gamma}$ for the case of Received Signal Strength measurements,  $f(u_t(k),u_t(m)) = \tan^{-1}\frac{Im(u_t(k)-u_t(m))}{Re(u_t(k)-u_t(m))}$ for Angle of Arrival measurements etc. Each measurement is modeled as either a LOS dominated signal or a NLOS dominated signal by choosing the observation noise in the received signal $n_t(km)$, to be drawn from a mixture of two distributions, $(p_{LOS}(\theta_t(km)|u_t(k),u_t(m)),p_{NLOS}(\theta_t(km)|u_t(k),u_t(m)))),$ with mixture probabilities $(\alpha,1-\alpha)$ respectively. The model is motivated by some of the experimental work  carried out in the UWB  \cite{renzo2006ultra, turin1972statistical, pedersen2000stochastic} which validate that {\em some} fraction of the received signals are purely LOS-dominated signals. Let $z_t(km)$ be an indicator random variable which is 1 if  $\theta_t(km) \sim p_{LOS}(\theta_t(km)|u_t(k),u_t(m)),$ 0 otherwise. $z_t(km)$ can be correlated in time (e.g. an obstruction between two vehicles could lead to sustained NLOS condition) which we capture using a Markov chain with stationary distribution $(\alpha,1-\alpha)$. Let $p(u_t(k)|u_{t-1}(k))$ be the distribution that governs the evolution of the vehicle states across time which is obtained based on the inertial navigation system measurements. \\
  
 The goal is for each vehicle $k$   to estimate its own location $u_t(k)$, based on all measurements $\{\theta_{\tau}(km)\}_{\tau = 1}^t$ from its neighbors upto time $t$.  Given the non-gaussian nature of the problem, we adopt a \emph{particle filtering} approach that is a popular Monte-Carlo technique which can provide accuracies close to Minimum Mean Square Error (MMSE) estimates. The nature of our problem helps us obtain Kalman-like updates for particle filtering giving rise to a simplified and practical algorithm. We describe the graphical model formulation and our algorithm in the next section.
  
\section{Graphical models \& Particle Filtering for Cooperative Localization}
\label{sec:GraphModel}
\begin{figure}
\centering
\includegraphics[height=2.5in]{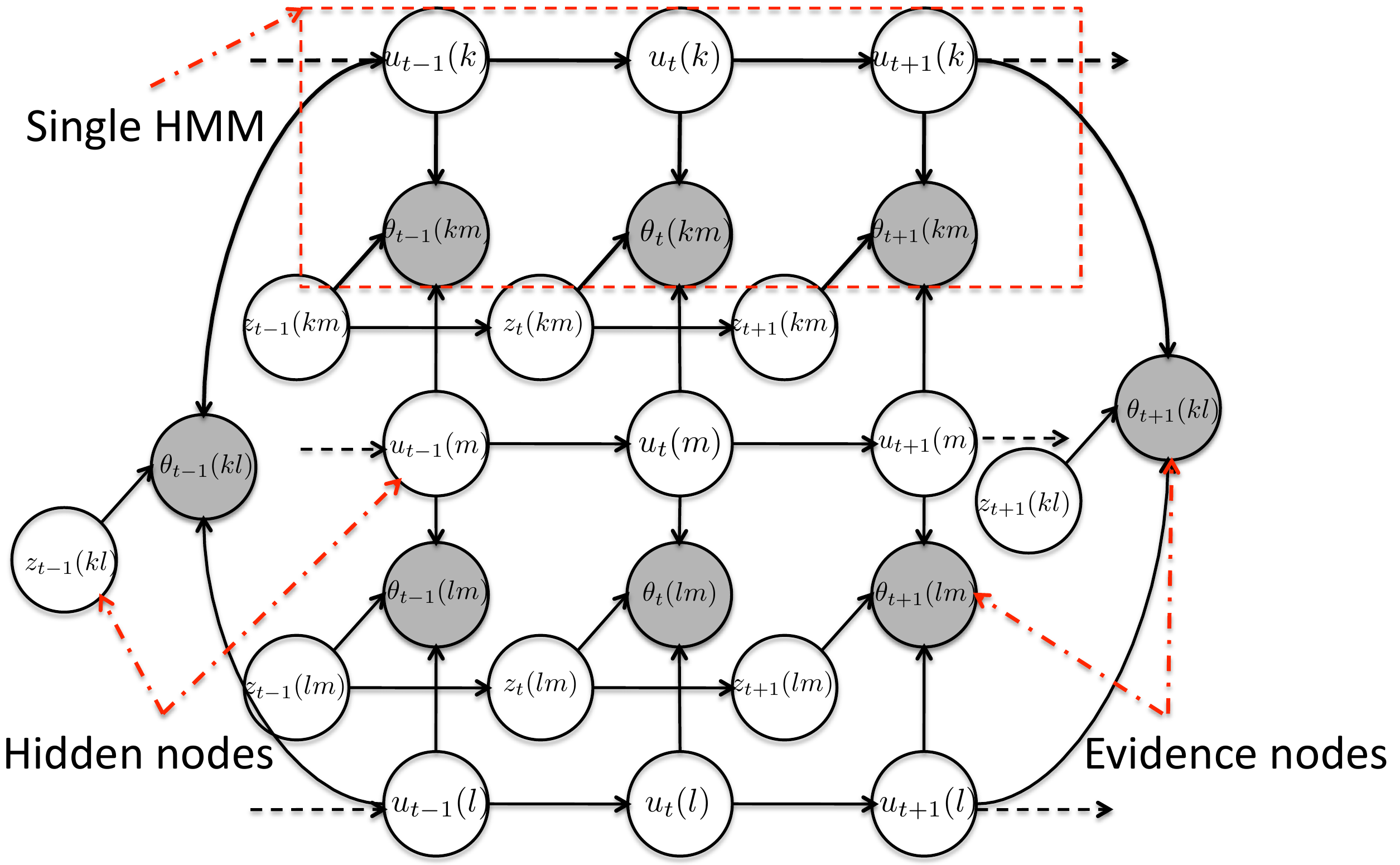}
\caption{Graphical model representation of the unknown vehicle states (locations) and the observations (LOS/NLOS measurements).}
\label{fig:VehHmm}
\vspace*{0in} \end{figure}
Graphical models provide a good way of modeling dependencies between different random variables  \cite{jordan1998learning}  and thereby exploit the probability structure to develop computationally feasible estimation algorithms. A directed acyclic graphical model, $\cal{G}(V,E)$, consists of a vertex set ${\cal V}$ where each vertex is associated with a random variable and an edge set ${\cal E}$ which is the collection of all directed edges, with the conditional independencies encapsulated by graph separation \cite{jordan1998learning}. Dependencies amongst the vehicle locations and the readings are captured by the graphical model shown in Fig \ref{fig:VehHmm}, which is a coupled Hidden Markov Model (HMM). The unshaded nodes are the \emph{hidden nodes}  to be estimated \footnote{Anchor nodes are not shown in this model for simplicity.} and the shaded nodes are observations coupling the different Markov chains of the vehicles, termed as \emph{evidence nodes}.  The joint probability distribution of all the random variables is given by,
\bea
p(\{u_t\}, \{\theta_t\}, \{z_t\} ) & = & \prod_{t,k,m}p(\theta_t(km)|u_t(k),u_t(m),z_t(km)) \prod_{k=1}^N p(u_1(k))\prod_t p(u_t(k)|u_{t-1}(k)
\eea
\bea 
& &   \prod_{k,m} p(z_1(km))\prod_t p(z_t(km)|z_{t-1}(km)). 
\eea 

Celebrated algorithms like the loopy belief propagation \cite{jordan1998learning} are hard to apply here given the continuous value of the state space and the dynamic evolving nature of the graph.

\begin{figure}
\centering
\includegraphics[height=1.5in]{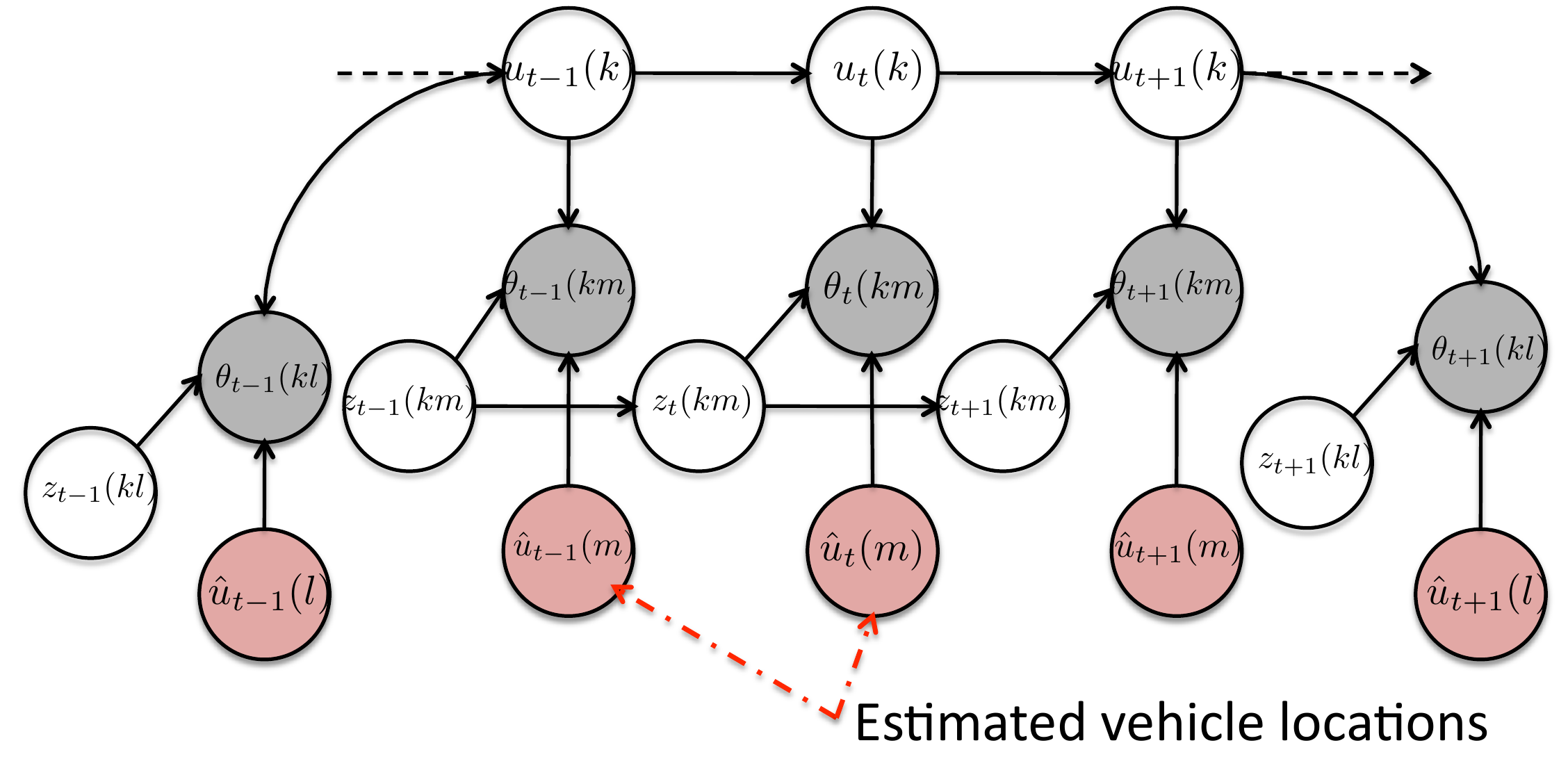}
\caption{Approximation for the original coupled HMM.}
\label{fig:VehHmmApp}
\vspace*{0in} \end{figure}
Particle Filtering is a well-known Monte Carlo simulation technique  \cite{doucet2001sequential} whose goal is to approximate
the posterior state density that can be used to obtain the MMSE estimate. Exact inference over  Fig \ref{fig:VehHmm} being hard,  we resort to an approximation for every vehicle as shown in Fig \ref{fig:VehHmmApp}. The new set of shaded nodes in this graph correspond to the estimated location of the other vehicles. At every time instant, each vehicle gets the estimated location of its neighbors from the previous time instants and assuming that it is close enough to the true location, the vehicle gets an estimate of its own location using particle filtering. This also reduces the communication overhead between vehicles and each vehicle needs to only transmit its location estimate and possibly the variance. Treating $\{u_t,z_t\}$ as a single state would lead to a exponentially large state space with increasing neighbors. We describe an efficient algorithm to circumvent this issue.\\

Consider only a single vehicle for now and omit the vehicle index $k$ for notational simplicity. Let $\{u_{1:t},z_{1:t}\}$ and $\theta_{1:t}$ denote
the set of states and observations up to time $t$ respectively. The posterior density of the vehicle location given the observations can be approximated as 
$p(u_t|\theta_{1:t}) \approx \hat{p}(u_t|\theta_{1:t}) = \frac{1}{K}\sum_{i=1}^K \delta(u_t - u_t^{i})$
where $\{u_{t}^{i} ; i=1,...,K\}$ are K i.i.d random samples (particles) picked from the distribution $p(u_t|\theta_{1:t})$ and $\delta(.)$ is the dirac-delta function. Given the hardness of sampling from $p(u_t|\theta_{1:t})$ , we choose a proposal distribution $\pi(u_{1:t}|\theta_{1:t})  = \pi(u_{1:t-1}|\theta_{1:t-1}) \pi(u_t|u_{1:t-1}, \theta_{1:t})$. Suppose that we are interested in estimating the mean of some function, $f(u_{1:t})$ of the vehicle location across time, i.e. $\mathbb{E}(f(u_{1:t}))$. We then have the following:
\bea
\mathbb{E}(f(u_{1:t})) & = & \int f(u_{1:t}) \frac{p(\theta_{1:t}|u_{1:t}) p(u_{1:t})}{\pi(u_{1:t}|\theta_{1:t}) p(\theta_{1:t})} \pi(u_{1:t}|\theta_{1:t}) du_{1:t},\\
& = & \frac{\mathbb{E}_{\pi} (f(u_{1:t}) w_t(u_{1:t}))}{\mathbb{E}_{\pi}(w_t(u_{1:t}))},
\eea
where $w_t(u_{1:t}) = \frac{p(\theta_{1:t}|u_{1:t}) p(u_{1:t})}{\pi(u_{1:t}|\theta_{1:t})} = \sum_{z_t} \frac{p(\theta_{1:t}, z_t|u_{1:t}) p(u_{1:t})}{\pi(u_{1:t}|\theta_{1:t})}$. Let, 
\bea \phi(u_{1:t},z_t)  =  \frac{p(\theta_{1:t}, z_t|u_{1:t})  p(u_{1:t})}{\pi(u_{1:t}|\theta_{1:t})}. \eea
We can now obtain a recursive update equation for $\phi(.)$ as follows.
\bea
 \phi(u_{1:t},z_t)  & = & \frac{p(\theta_{1:t}, z_t|u_{1:t})  p(u_{1:t})}{\pi(u_{1:t}|\theta_{1:t})},\\
                              & = & \sum_{z_{t-1}}  \frac{p(\theta_{1:t}, z_t, z_{t-1}|u_{1:t})  p(u_{1:t})}{\pi(u_{1:t}|\theta_{1:t})},\\
                              & = & \sum_{z_{t-1}} \frac{p(z_t,z_{t-1}|u_{1:t}) p(\theta_{1:t-1}|z_{t-1},u_{1:t-1})p(\theta_t|z_t,u_t)}{\pi(u_{1:t}|\theta_{1:t})}\\
                              & = & \sum_{z_{t-1}} \frac{p(z_t,z_{t-1})p(\theta_{1:t-1}|z_{t-1},u_{1:t-1})p(\theta_t|z_t,u_t)}{\pi(u_{1:t}|\theta_{1:t})}\\
                              & = & \frac{p(\theta_{t}|z_t,u_t) p(u_t | u_{t-1})}{\pi(u_t|u_{1:t-1} \theta_{1:t})} 
                             \sum_{z_{t-1}} p(z_t|z_{t-1}) \phi(u_{1:t-1},z_{t-1}),
\eea
where the last three equalities follow from the conditional independence structure of the different random variables (i.e. $z_t$'s are unconditionally independent of $u_t$'s and $\theta_t$ only depends on $z_t$ and $u_t$).

Choosing $\pi(u_t|u_{1:t-1} \theta_{1:t}) = p(u_t | u_{t-1})$, we get,
\bea \phi(u_{1:t},z_t)  = p(\theta_{t}|z_t,u_t) \sum_{z_{t-1}} p(z_t|z_{t-1}) \phi(u_{1:t-1},z_{t-1}).\eea
For each vehicle $k$ and its neighbor $m$ we have $\phi_t(km)$ and the update equations are given in Algorithm 1. All the expectations in the above equations are replaced by summations over samples of $u_t^i$ taken from the chosen proposal distribution. To take care of degeneracy issues over long time instants \cite{doucet2001sequential}, we employ the standard resampling procedure whenever the number of distinct particles fall below a threshold. In our simulations in Section \ref{sec:sim2}, we had 2000 particles and whenever the effective sample size calculated as $(\hat{N}_{eff} = \frac{1}{\sum_i (w_t^i)^2})$ was below a threshold (30 here), a simple multinomial resampling was carried out. If the number of distinct particles was too low, then all the particles were newly sampled around the estimated position.\\
 
  Conceptually, the performance of the algorithm depends on different system parameters such as the noise in the ranging measurements, INS noise, number of anchors, number of neighbors etc. In particular. the algorithm banks on the INS reading to be good enough to sustain the vehicle location until enough LOS measurements have been obtained. The algorithm is more sensitive to INS noise than typical belief propagation algorithms given that each vehicle extrapolates the neighbor's location estimate to the current time instant using the INS reading and there is no iteration to improve this estimate at that time instant. Hence large errors in INS readings would lead to significant performance degradation. The hope is that the INS measurements in existing vehicles can be quite precise for the algorithm to work well \cite{rezaei2007kalman}. Further the LOS measurements are assumed to be quite precise and the number of neighbors is assumed to be large and well spread. Failure of any or some of these conditions would lead to performance degradations.  Extensive simulation results under a realistic setup are presented in Section \ref{sec:sim2}.

   \begin{figure}
 \centering
 \includegraphics[height=2.5in]{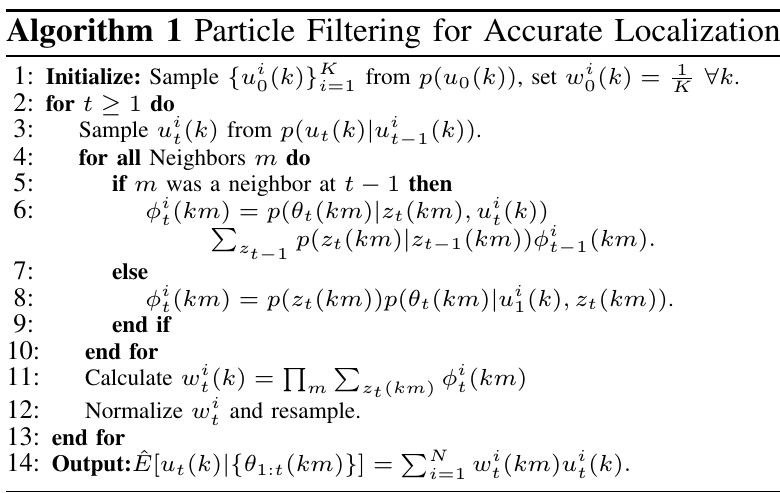}
 \vspace*{0in} \end{figure}

\section{Scaling Laws for Cooperative NLOS localization} 
From an engineering design perspective, it is important to understand the effect of the node geometry, connectivity graph, number of agents and anchors and the fraction of LOS measurements, on the localization performance. We modestly attempt to analyze these effects in a simple static setup with $N$ agents and $M$ anchors placed in a two-dimensional field. Let ${\bf u} = \{u(k)\}$ is the vector of all agent locations and ${\bf \hat{u}}$ the estimated locations. We study the the behavior of  the Cramer-Rao Lower Bound (CRLB), which is a lower bound on the estimation error for the class of {\em unbiased } estimators  $\displaystyle (\mbox{i.e.} \ \ \mathbb{E}(\hat{{\bf u}}) = {\bf u}, \ \mbox{where} \ \mathbb{E} $ is the expectation operator).  Let ${\bf \eta} = \left[
\begin{array}{rl} {\bf u}_R ; {\bf u}_I\end{array}\right]$\footnote{We use the notation $[{\bf a};{\bf b}]$ to represent the vertical concatenation of two column vectors ${\bf a}$ and ${\bf b}$. }, where ${\bf u_R} = Re\{{\bf u}\}, {\bf u_I} = Im\{{\bf u}\}$, be the vector of  parameters to be estimated. The Cramer-Rao theorem states that, 
$\mathbb{E}[({\bf u}-{\bf \hat{u}})({\bf u}-{\bf \hat{u}})^*] \succeq F^{-1},$
where ${\bf a}^*$ represents conjugate transpose of a complex column vector ${\bf a}$ and the matrix $F$, known as the Fisher Information Matrix, is defined by,
$F_{km} \triangleq \mathbb{E}\left\{  \frac{\partial \ln p({\Theta | \eta})}{\partial \eta_k}   \frac{\partial \ln p({\Theta | \eta})}{\partial \eta_m} \right\}.$

\begin{figure}
\centering
\includegraphics[height = 2in]{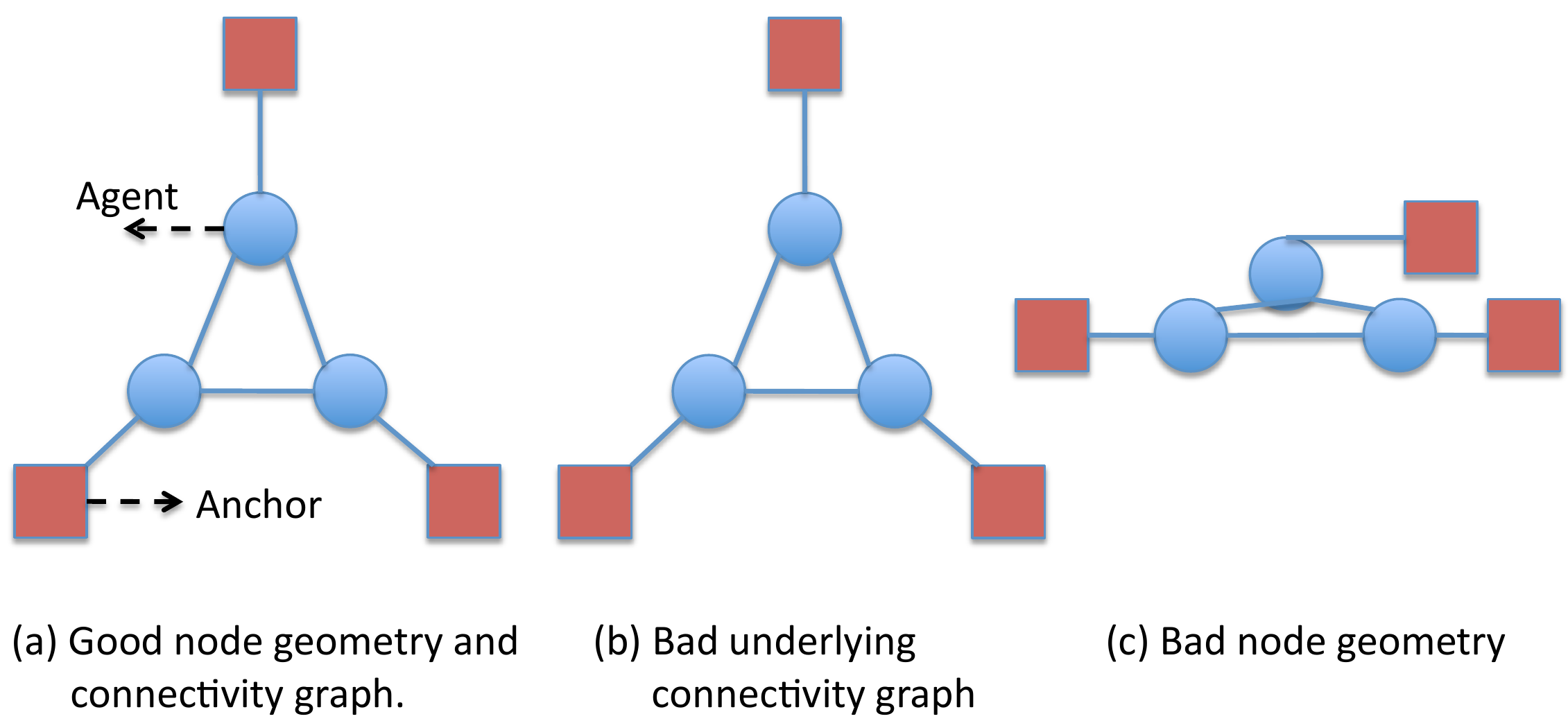}
\caption{Examples of node geometries and connectivity graphs.}
\label{fig:NodeGeometry}
\vspace*{0in} \end{figure}
The errors in localization can be attributed to two factors, one being the relative geometry of the anchors and vehicle locations (see Fig. \ref{fig:NodeGeometry}) and the other being the noise in the measurements. Thus one can naturally ask the question, ``Can the effect of node geometry and that of noise be independently analyzed?''. We answer this question in the affirmative under certain simplifying assumptions of homogenous noise statistics. Our theory is developed for the measurement model in Section \ref{sec:ProbSet}, with the further assumption that all the readings are independent and identically distributed. Even though for a fixed realization, each vehicle might experience different fading environments, when averaged across multiple deployments of the vehicles and scatterers and time, it is reasonable to assume that the statistics of the noise is similar across the network. Further, this simplifying assumption helps us gain intuition on how the different system parameters affect the performance of the system.
 The following theorem states the separation under the above conditions.
\begin{thm}[{Separation Principle:}]
The Fisher Information Matrix can be written as $F = g(p_{NOISE})F_G$, where the matrix $F_G$ depends only on the node locations and the underlying connectivity graph. The scalar function $g(.)$ is given by,
$g(p_{NOISE}) = \mathbb{E} \left\{\left(\frac{ \partial }{ \partial n} \mbox{ \em ln } p_{NOISE}({n})\right)^2\right\},$
under the assumption that $p_{NOISE}$ is differentiable over its support $[LL,UL]$ and $p(UL) - p(LL) = 0$.
\end{thm}
\begin{proof}
For simplicity, let us focus on the case where $\theta(km) =  ||u(k)-u(m)|| + n(km)$, though the derivation extends to the general case.  The $(k,m)$th entry in the matrix is given by,
$F_{km} \triangleq \mathbb{E}\left\{  \frac{\partial \ln p({\Theta | \eta})}{\partial \eta_k}   \frac{\partial \ln p({\Theta | \eta})}{\partial \eta_m} \right\}.$
Let ${\cal N}(k)$ be the set of neighbors of node $k$. Let us focus on the case when $\eta_k = u_{R}(k)$ and $\eta_m = u_{R}(m)$.  It is easy to show that,
\bea
\frac{\partial \ln p({\Theta | \eta})}{\partial \eta_k} = \sum_{i \in {\cal N}(k)}  \frac{\partial \ln p_{NOISE}({n(ki)})}{\partial n(ki)} \cos(\phi_{ki}), 
\eea
where  $n(ki) = \theta(ki) - ||u(k) - u(i)||$ and  $\cos(\phi_{ki}) =  \frac{(u_{R}(k) - u_{R}(i))}{||u(k) - u(i)||}$ (note that $\cos(\phi_{ki}) = -\cos(\phi_{ik})$ ). Thus we have,
\bea
F_{km} & = & \mathbb{E} \left \{ \sum_{i \in {\cal N}(k)}  \frac{\partial \ln p({n(ki)})}{\partial n(ki)} \cos(\phi_{ki}) \sum_{j \in {\cal N}(m)}  \frac{\partial \ln p({n(mj)})}{\partial n(mj)} \cos(\phi_{mj})\right\}, \\
 & = & \mathbb{E} \left \{ \sum_{i \in {\cal N}(k)}  \frac{ p'({n(ki)})}{ p(n(ki))} \cos(\phi_{ki}) \sum_{j \in {\cal N}(m)}  \frac{ p'({n(mj)})}{p(n(mj))} \cos(\phi_{mj})\right\}.
\eea
Consider the case when $k = m$. We get that,
 \bea
 F_{kk} & = & \mathbb{E} \left\{\sum_{i \in {\cal N}(k)} \cos^2(\phi_{ki}) \left(\frac{ p'({n(ki)})}{ p(n(ki))}\right)^2 +  
 \sum_{(i \neq j)\in {\cal N}(k)} \cos(\phi_{ki}) \cos(\phi_{kj}) \frac{ p'({n(ki)})}{ p(n(ki))} \frac{ p'({n(kj)})}{ p(n(kj))} \right\}, \\
 & = & g(p_{NOISE})\sum_{i \in {\cal N}(k)} \cos^2(\phi_{ki}), 
 \eea
where $g(p_{NOISE}) = \mathbb{E} \left\{\left(\frac{ p'({n(ki)})}{ p(n(ki))}\right)^2\right\}$. Assuming $p_{NOISE}(n)$ to be differentiable on its support $[LL,UL]$ and that $p(UL) = p(LL)$, we have,$\mathbb{E} \left\{\frac{ p'({n(ki)})}{ p(n(ki))}\right\} = 0$.

Similarly one can simplify the equations for different combinations of $k$ and $m$ and $\eta_k$ and $\eta_m$ to obtain the result $F = g(p_{NOISE}) F_G$, where the entries are given as follows (see  \cite{ekambaram2012cooploc}).

 For $\eta_k = u_{R_k}$ and $\eta_m = u_{R_m}$, we get
  \bea
  F_{km}& = & g(p_{NOISE})\left\{
\begin{array}{rl}  \displaystyle \sum_{i \in {\cal N}(k)} \cos^2(\phi_{ki}) & \mbox{if } k = m, 
                             \\ -\cos^2(\phi_{km}) & \mbox{if } m \in {\cal N} (k), 
                             \\ 0 & \mbox{o.w.}\end{array} \right. .  
\eea
 
 For $\eta_k = u_{I_k}$ and $\eta_m = u_{I_m}$ we have,
 \bea
  F_{km}& = & g(p_{NOISE})\left\{
\begin{array}{rl}  \displaystyle \sum_{i \in {\cal N}(k)} \sin^2(\phi_{ki}) & \mbox{if } k = m, 
                             \\ -\sin^2(\phi_{km}) & \mbox{if } m \in {\cal N} (k), 
                             \\ 0 & \mbox{o.w.}\end{array} \right. .  
\eea
 
 For $\eta_k = u_{R_k}$ and $\eta_m = u_{I_m}$ we have,
  \bea
  F_{km}& = & g(p_{NOISE})\left\{
\begin{array}{rl}  \displaystyle \sum_{i \in {\cal N}(k)} \sin(\phi_{ki}) \cos(\phi_{ki}) & \mbox{if } k = m, 
                             \\ -\sin(\phi_{km})\cos(\phi_{km}) & \mbox{if } m \in {\cal N} (k), 
                             \\ 0 & \mbox{o.w.}\end{array} \right. .  
\eea

The $g(p_{NOISE})$ term can be factored out to get the desired result.
\end{proof}
For a mixture distribution, we have  $
 g(p_{NOISE})  =  \int_{-\infty}^{+\infty} \frac{ (\alpha p'_{LOS}(n) + (1-\alpha)p'_{NLOS}(n))^2}{ (\alpha p_{LOS}(n) + (1-\alpha)p_{NLOS}(n))} dn.$ The assumption on $p_{NOISE}$ holds for a wide class of distributions such as  ex-gaussian, gaussian mixtures, uniform distribution etc, that are commonly used models in the NLOS setup. The assumption on the identical statistics of the readings is mainly questionable when we consider measurements from anchors like GPS and measurements from other neighboring vehicles which can have widely different noise parameters.  In such a case, one can easily show that the Fisher Matrix can be written as $F = g(p_{NOISE}^{veh})F_G^{veh} + g(p_{NOISE}^{sat})F_G^{sat}$, wherein $F_G^{veh}$ is the Fisher Matrix derived by only using the inter-vehicular measurements and $F_G^{sat}$ is derived by only considering the satellite-vehicle measurements.
 
 \begin{figure}
 \centering
 \subfigure[Variation of $\frac{1}{g(p_{NOISE})}$ as a function of $\alpha$ for ex-gaussian and Gaussian mixture distributions.]
 {\includegraphics[height = 2.5in]{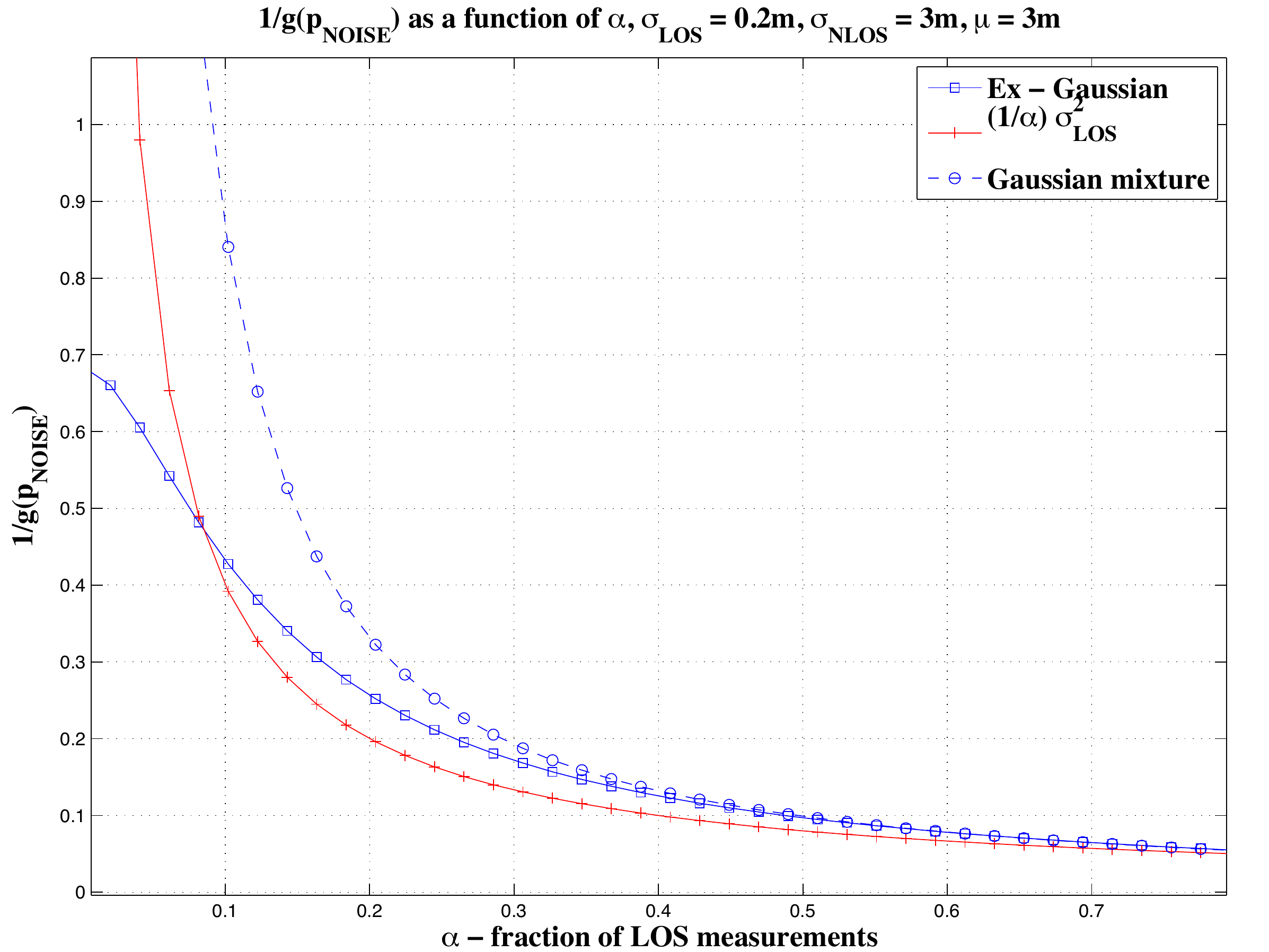} \label{fig:crlb_nlos_alpha}}
 \subfigure[$g'(\alpha)$ as a function of $\alpha$ for ex-gaussian.]
 {\includegraphics[height = 2.5in]{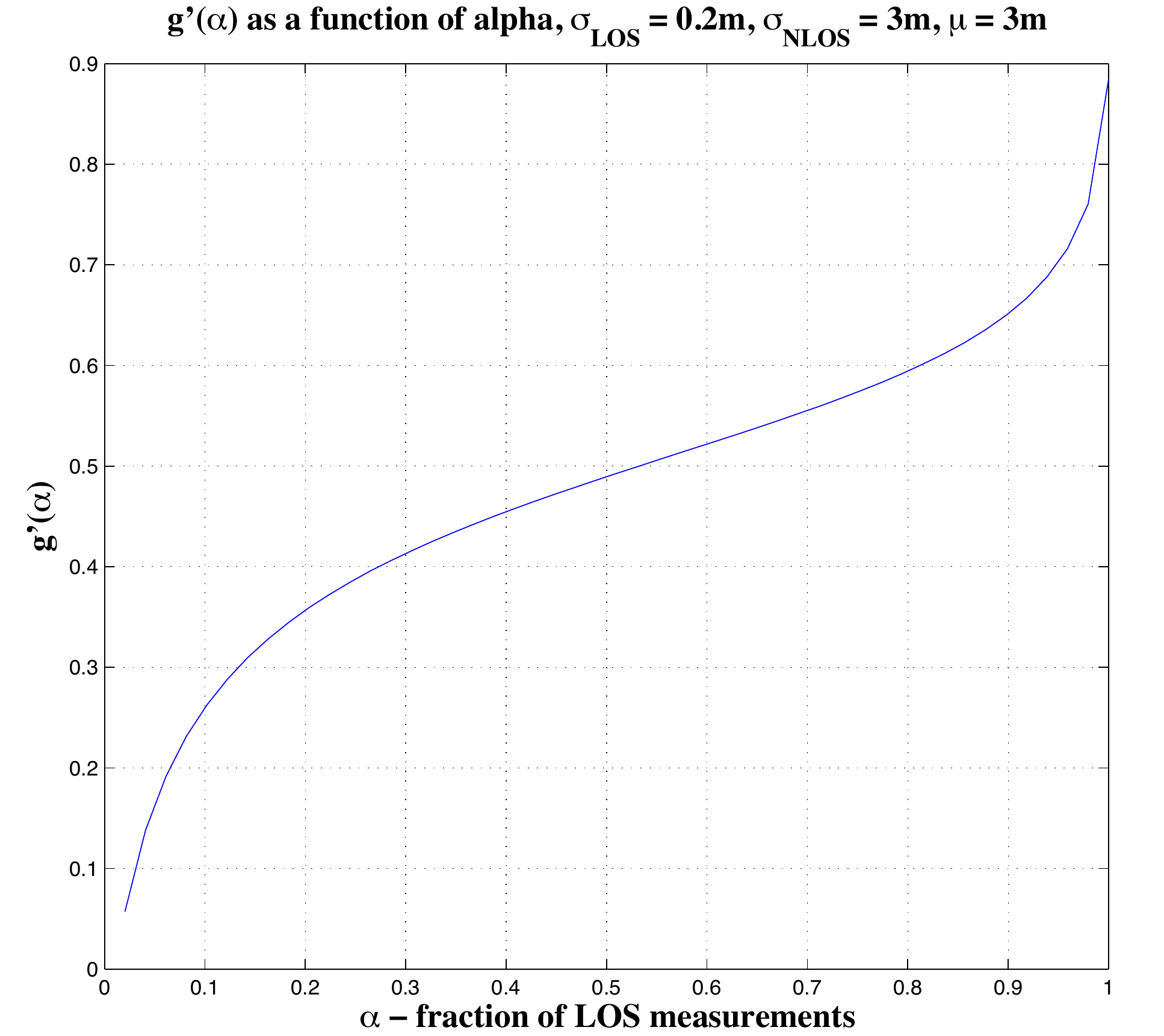}\label{fig:crlb_nlos_diffalpha}}
 \caption{Behavior of the CRLB as a function of the different noise parameters.}
 \vspace*{0in} \end{figure}
 
%  \begin{figure}
%\includegraphics[height = 2.5in]{CRLB_NLOS_alpha}
%\caption{Variation of $\frac{1}{g(p_{NOISE})}$ as a function of $\alpha$ for ex-gaussian and Gaussian mixture distributions.}
%\label{fig:crlb_nlos_alpha}
%\vspace*{0in} \end{figure}
% \begin{figure}
%   \hspace*{-0.2in}
%\includegraphics[height = 3in]{CRLB_NLOS_gdiffalpha}
%\caption{$g'(\alpha)$ as a function of $\alpha$ for ex-gaussian.}
%\label{fig:crlb_nlos_diffalpha}
%\vspace*{0in} \end{figure}
% \begin{figure}
% \centering
%\includegraphics[height = 2.6in]{CRLB_NLOS_mu}
%\caption{Variation of $\frac{1}{g(p_{NOISE})}$ as a function of the exponential mean parameter $\mu$ for the ex-gaussian and Gaussian mixture distributions.}
%\label{fig:crlb_nlos_mu}
%\vspace*{0in} \end{figure}

To analyze the behavior, we focus on distance measurements and consider a simple mixture model for the noise distribution. We will assume that $p_{LOS} \sim {\cal N}(0,\sigma_{LOS}^2)$. Fig. \ref{fig:crlb_nlos_alpha}, shows the behavior of $g(p_{NOISE})$ as a function of $\alpha$ for $p_{NLOS}$ being ex-gaussian and  positive mean gaussian having the same mean and variance. The scaling is compared against $\frac{\sigma^2}{\alpha}$ which seems to be a good approximation of the behavior at higher values of $\alpha$. At lower $\alpha$ values, $\frac{\sigma^2}{\alpha}$ would be worse off given that we are completely discarding the NLOS measurements. However, the ex-gaussian curve performs better given that we make use of the noise statistics. Fig \ref{fig:crlb_nlos_diffalpha} shows a plot of $g'(\alpha) = \frac{\partial}{\partial \alpha} g(.)$. The interesting trend  to note here is that the differential change in $g(.)$ varies sharply at lower values and higher values of $\alpha$. This suggests that at lower values of $\alpha$, even a small fraction of LOS measurements can significantly help and at higher $\alpha$ a small fraction of NLOS can significantly hurt.\\ 

To analyze the effect of the number of nodes on the localization error, we focus on the case, $\theta(km) = ||u(k)-u(m)||$.
We undertake an {\em incremental analysis} by starting out with an existing network of $N$ agents and $M$ anchors with Fisher Matrix $F = g(p_{NOISE})F_G$ and quantifying the effect of adding additional anchors and agents in the network.  We assume that the node density is large enough and ignore any boundary effects in our derivations. Suppose we randomly deploy additional $\tilde{M}$ new anchors and $\tilde{F}$ is the new Fisher Matrix, we have that,

\begin{thm}[{ Scaling law for anchors:}]
For large $\tilde{M}$, the lower bound on the sum mean squared localization error is given by,
$ \mbox{Trace} (\tilde{F}^{-1}) = \frac{1}{g(p_{NOISE})}\displaystyle \sum_{i=1}^{2N} \frac{1}{\lambda_i + \frac{\rho \tilde{M}}{2}},$
where $\rho= \pi R^2$ with $\rho \tilde{M}$ being the average number of new anchor measurements of every node and $\lambda_i$'s are the eigen values of $F_G$ (assumed to be full rank).
\end{thm}
\begin{proof}
We provide a sketch of the proof here with a simple example and defer the detailed proof to Appendix A.
 Consider the case where we add a single additional anchor $a$ in a network of $N$ agents and $M$ anchors. Let $\tilde{F}$ be the new Fisher Information Matrix. Then it is easy to see that $\tilde{F}$ is the same as $F$ except for the following entries. If node $k$ is in the communication radius of $a$, then we have, for $\eta_k = u_{R_k}$, $\tilde{F}_{kk} = F_{kk} + \cos^2(\phi_{ka})$,  for $\eta_k = u_{I_k}$, $\tilde{F}_{kk} = F_{kk} + \sin^2(\phi_{ka})$ and for $\eta_k = u_{R_k}, \eta_m = u_{I_k}$, $\tilde{F}_{km} = F_{km} + \cos(\phi_{ka})\sin(\phi_{ka})$ where $\phi_{ka}$ is the angle between node $k$ and anchor $a$. For all other entries,  $\tilde{F}_{kk} = F_{kk}$. Thus when $\tilde{M}$ anchors are randomly added into the network, the new Fisher Information Matrix has the following structure.

\bea
\small
\tilde{F}  =  F +  \left[ \begin{array}{ccc|ccc} % brackets may be (...), [...], \{...\}, or left out
       \displaystyle\sum_{a \in {\cal N}(1)} \cos^2(\phi_{1a}) & 0 &0& \displaystyle\sum_{a \in {\cal N}(1)} \cos(\phi_{1a}) \sin(\phi_{1a}) & 0 &0\\
       0 & ...& 0 &  0 & ...& 0\\
       0 & 0 &  \displaystyle\sum_{a \in {\cal N}(N)} \cos^2(\phi_{Na}) &0 & 0 &  \displaystyle\sum_{a \in {\cal N}(N)} \cos(\phi_{Na})\sin(\phi_{Na}) \\
       \hline
        \displaystyle\sum_{a \in {\cal N}(1)} \cos(\phi_{1a}) \sin(\phi_{1a})& 0 &0& \displaystyle\sum_{a \in {\cal N}(1)}  \sin^2(\phi_{1a}) & 0 &0\\
       0 & ...& 0 &  0 & ...& 0\\
       0 & 0 &  \displaystyle\sum_{a \in {\cal N}(N)} \cos(\phi_{Na})  \sin(\phi_{Na}) &0 & 0 &  \displaystyle\sum_{a \in {\cal N}(N)} \sin^2(\phi_{Na}) \\
    \end{array}\right] ,
\eea

Since the anchors are randomly placed, it is reasonable to assume that the angles $\phi_{ka}$ are i.i.d. $U(0,2\pi)$ for each node $k$. Thus it is easy to see that for large $\tilde{M}$ the above expression converges to $\tilde{F} \rightarrow F + \frac{\rho \tilde{M}}{2} I_{2N} $, where $\rho \tilde{M}$ is the average number of newly added anchors that are neighbors of each node assuming a homogenous placement of the nodes. Thus $\mbox{Trace}(\tilde{F}^{-1}) = \sum_{i=1}^{2N}  \frac{1}{\lambda_i + \frac{\rho \tilde{M}}{2}}$, where $\lambda_i$ are the eigen values of $F$.
\end{proof}
 The eigen values of the Fisher Information Matrix,   $\{\lambda_i\}_{i=1}^{2N}$ can be interpreted as a measure of the ``precision'' in the agent location estimates. The factor $\frac{\rho \tilde{M}}{2}$ is the ``additional precision''   from the newly added anchors. 
 If under the same setup we add $\tilde{N}$ new agents which only get measurements from the existing nodes in the network, then,

\begin{thm} [{Scaling law for agents:}]
For large $\tilde{N}$, the Fisher Information Matrix $\tilde{F}$  is given by,\\
 \bea
 \tilde{F} & = & F + \frac{1}{g(p_{NOISE})}\frac{\rho \tilde{N}}{2}\left( \left(1-\frac{1}{\rho(N+M)}\right)I_{2N} \right. \left. - \frac{1}{\rho(N+M)}\left[
\begin{array}{cc}
  {\bf 1}_N{\bf 1}_{N}^T&  {\bf 0}  \\
 {\bf 0} &    {\bf 1}_N{\bf 1}_{N}^T \\
\end{array}
\right]\right).
\eea
\end{thm}

\begin{proof}
As before we provide a sketch of the proof with an example and defer the detailed proof to Appendix B. Consider the case when $\tilde{N}$ agents are added to an existing network of $N$ agents and $M$ anchors. We will assume that the newly added nodes get measurements {\em only} with the existing $N$ nodes in the system that are within their communication radius. This simplifying assumption helps get a better understanding of the scaling behavior. We will provide an intuition for the derivation using an example here and refer the reader to \cite{ekambaram2012cooploc} for more details. Let us suppose that we add a single additional agent to the existing network of agents and anchors. Let the location of this agent be $z_{R1} + j z_{I1}$. Further assume that this agent gets measurements from all the existing agents in the network. Define the new parameter vector as $\tilde{\eta} = [{\bf u}_R;{\bf u}_I;z_{R1};z_{I1}]$. One can verify that the new Fisher Information Matrix, $\tilde{F}$ has the following form,\\
\begin{figure}[h!]
\centering
\includegraphics[height = 2.5in]{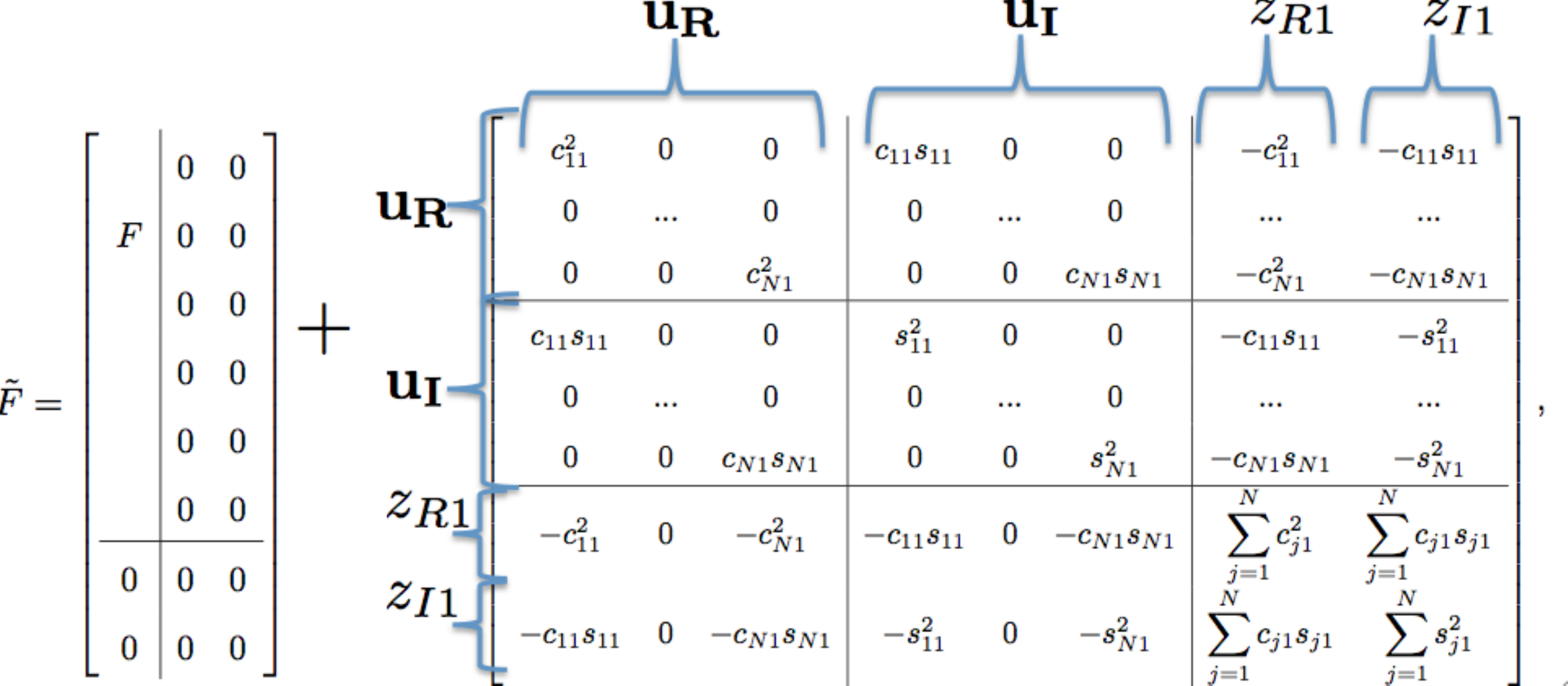}
\end{figure}
\vspace*{-0.9in}
\bea
\eea
   where $c_{j1} = \cos(\tilde{\phi}_{j1})$ and $s_{j1} = \sin(\tilde{\phi}_{j1})$ and $\tilde{\phi}_{j1}$ is the angle between the newly added agent and node $j$. The $\tilde \ $ notation is used to distinguish the angles formed due to the newly added agents as opposed to the existing agents in the system.
  If the N nodes are randomly placed, then the angles are $U(0,2\pi)$. Thus one can see that the following entries converge by strong law of large numbers  for large $N$, i.e.,  $\sum_{j=1}^N\cos^2(\phi_{j1}) \rightarrow \frac{N}{2}$, $\sum_{j=1}^N\sin^2(\phi_{j1}) \rightarrow \frac{N}{2}$ and $\sum_{j=1}^N\cos(\phi_{j1}) \sin(\phi_{j1})\rightarrow 0$. Now let us suppose that we add another agent with location $z_2$ in the network. In order to look at positioning error of the first $N$ nodes, we need to look at the Schur Complement of the leftmost $N \times N$ submatrix of $\tilde{F}$. The Schur complement of this submatrix is given by $F + G- \frac{2}{N}BB^T$, where for the parameter ordering $\tilde{\eta} = [{\bf u}_R;{\bf u}_I;z_{R1};z_{R2};z_{I1};z_{I2}]$, the matrices $B$ and $G$ are given by,
   \bea
   \small
   B =  \left[ \begin{array}{cccc} % brackets may be (...), [...], \{...\}, or left out
-c_{11}^2 & -c_{12}^2 & -c_{11}s_{11} & -c_{12}s_{12}\\
...& ...& ... & ...\\
 -c_{N1}^2 &  -c_{N2}^2 &-c_{N1} s_{N1}&-c_{N2} s_{N2}\\ 
  -c_{11}s_{11}&  -c_{12}s_{12}& -s_{11}^2 & -s_{12}^2\\
   ...&...&...&...\\
-c_{N1}s_{N1} & -c_{N2}s_{N2} & -s_{N1}^2 & -s_{N2}^2
   \end{array}\right] , \eea 
     \bea
     \small
   G =    \left[ \begin{array}{ccc|ccc} % brackets may be (...), [...], \{...\}, or left out
c_{11}^2 +c_{12}^2        &0 & 0                       & c_{11} s_{11} +  c_{12} s_{12} & 0 & 0                      \\
 0                     & ...& 0                       &  0                     & ... & 0                       \\
 0                     &0  & c_{N1}^2 + c_{N2}^2       &  0                     & 0  & c_{N1}s_{N1} +  c_{N2}s_{N2} \\ 
 \hline
c_{11}s_{11}+c_{12}s_{12} &0 & 0                        & s_{11}^2  +s_{12}^2      & 0 & 0                      \\
 0                    & ...& 0                        &  0                     & ... & 0                       \\
 0                    &0  & c_{N1}s_{N1}+c_{N2}s_{N2} &  0                     & 0  & s_{N1}^2 + s_{N2}^2          \\ 
 \end{array}\right] .
   \eea
   
   Thus when we add $\tilde{N}$ additional nodes into the system, assuming that the nodes are uniform randomly placed, $G \rightarrow \frac{\tilde{N}}{2} I_{2N}$, $BB^T \rightarrow \frac{\tilde{N}}{4} I_{2N} + \frac{\tilde{N}}{4}\left[
\begin{array}{cc}
  {\bf 1}_N{\bf 1}_{N}^T&  {\bf 0}  \\
 {\bf 0} &    {\bf 1}_N{\bf 1}_{N}^T \\
\end{array}
\right]$. 
Extending this to the case when each additional node has measurements only with a fraction $\rho$ of the existing nodes and anchors as well the above derivation holds with $N$ replaced by $\rho(N+M)$ and $\tilde{N}$ replaced by $\rho \tilde{N}$.   
\end{proof}

\begin{cor}
For large $N$, assuming $F$ to be full rank, the above bound reduces to, 
 $\mbox{Trace}(\tilde{F}^{-1}) = \frac{1}{g(p_{NOISE})}\sum_{i=1}^{2N} \frac{1}{\lambda_i + \frac{\rho \tilde{N}}{2}}.$
\end{cor}
The result quantifies the benefits of co-operation between agents. The agents can be interpreted as {\em virtual anchors} in the network that help localize the other nodes in the system.

Now let us consider the setup where the agents are mobile and we have some estimates of their velocities at each time instant from their INS readings. The vector of parameters is taken as follows  ${\bf \eta} = \left[
\begin{array}{rl} {\bf u}_{R_1};  {\bf u}_{I_1}; ....;{\bf u}_{R_T}; {\bf u}_{I_T} \end{array}\right]$, where ${\bf u_t}$ is the vector of vehicle locations at time $t$ and we consider $T$ time instants. Assuming that each reading $\theta_t(km)$ is sampled independently from the noise distribution $p_{NOISE}(.)$ and that the velocity measurement $s_t(m)$ is given by,
$s_t(m) = u_t(m) - u_{t-1}(m) + w_t(m),$
where $w_t(m) \sim {\cal N}(0,\sigma_{INS}^2)$,
 the following theorem holds true.
 
 \begin{thm}[{Mobile case:}]
The Fisher Information Matrix $F$ is given by 
$F = g(p_{NOISE}) F_G +\frac{1}{ \sigma_{INS}^2} F_{INS},$
under the assumption that $p_{NOISE}$ is differentiable over its support $[LL,UL]$ and $p(UL) - p(LL) = 0$.
 \end{thm}
 \begin{proof}
 Appendix C.
 \end{proof}
 $F_G$ is a matrix that only depends on the node locations at different time instants (see Appendix D) and $F_{INS}$ is a constant matrix. 
  
  \section{Simulation results}
  \label{sec:sim2}

  The simulation setup consisted of vehicles on a four lane road with lane width as 4m and lane length of 100m. Eight vehicles were simulated in each lane moving with an average velocity of $30mph$. Two adjacent lanes had vehicles moving in one direction and two others had vehicles moving in the opposite direction.  The vehicles do not change their lanes. Vehicles broadcast messages every 100ms (in accordance with the DSRC standard). Each message contains the vehicle position estimate and an estimate of the variance in its location estimate. In addition, vehicles also get ranging estimates with their neighbors which can be obtained using short probe packets that are typically a few symbols long. \\

   \begin{figure}
 \centering
 \includegraphics[height = 2.5in]{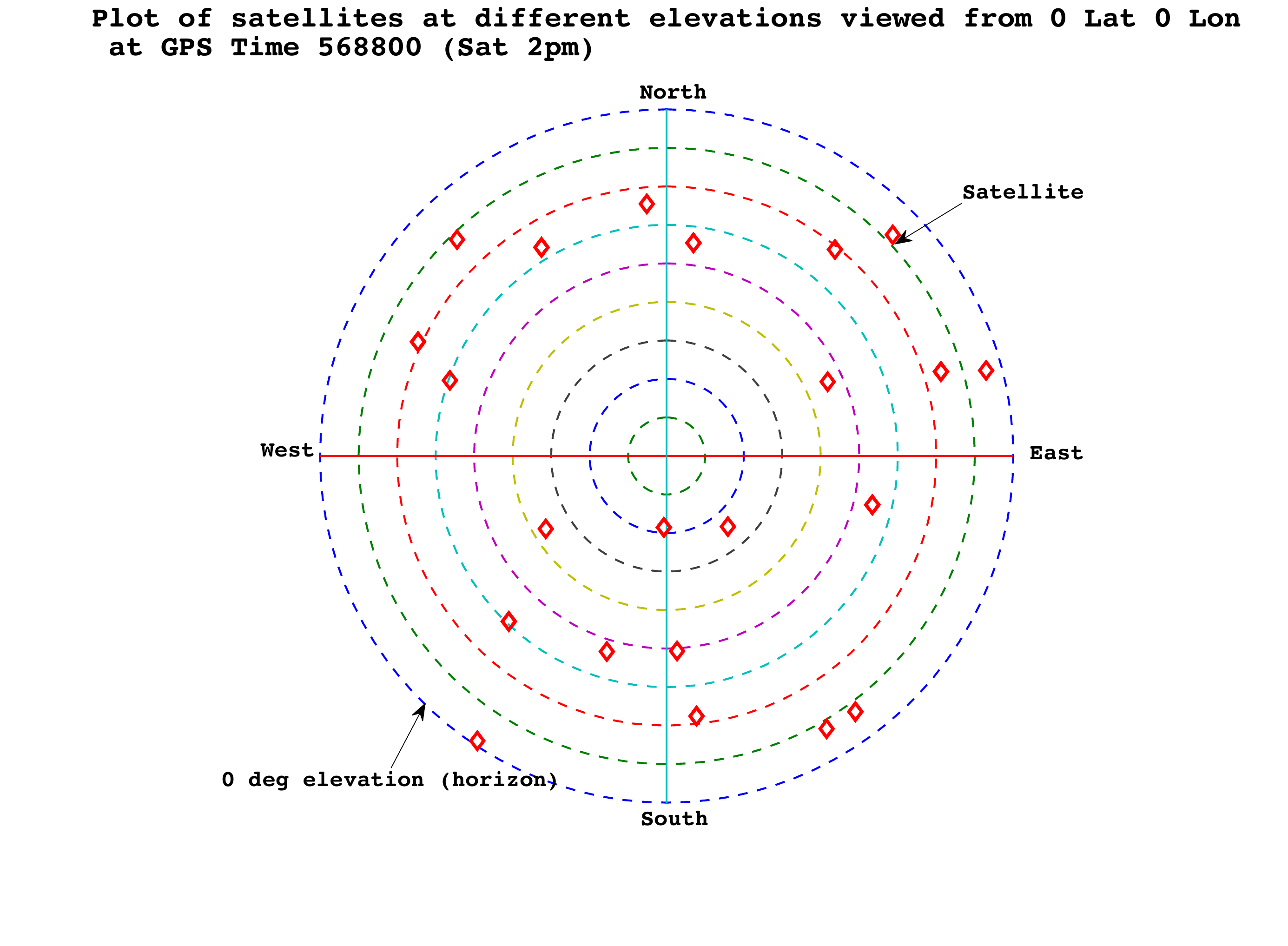}
 \caption{Satellite positions at different elevation angles viewed from 0 deg Lat and 0 deg Lon at GPS time 568800s (saturday 2pm). The different rings are increasing elevations of 10 degrees with the outermost ring corresponding to the horizon at 0 deg.}
 \label{fig:Skyplot}
 \vspace*{0in} \end{figure}  

Unless specified otherwise, the following parameters are used for the simulations. The radius of communication is assumed to be 50m (typically 50m-100m is used in these applications). GPS satellites are used as anchors. The mask angle for each vehicle at every time instant is taken to be uniform in the range 55 degrees to 85 degrees. This gives an average visibility of one satellite over every two time instants for each vehicle. Satellite locations are simulated from their almanac corresponding to the GPS time of week 568800s (UTC, saturday 2pm). Fig. \ref{fig:Skyplot}, shows the positions of different GPS satellites as seen from 0 deg Lat 0 deg Lon position for a 0 deg mask angle. The different rings correspond to different elevation angles spaced 10 degrees apart with the outermost ring being 0 degrees. The noise model is taken the be same as in section \ref{sec:GraphModel}. The LOS thermal noise standard deviation in the GPS readings is taken as $10m$ (no ionospheric/tropospheric biases assumed). The standard deviation in the NLOS noise is taken to be 5m to capture nearby reflectors such as  other vehicles and far away reflectors like buildings on the roadside. The INS noise is assumed to have a standard deviation of $1m/s$ \cite{rezaei2007kalman}. The algorithm was initialized with 2000 particles for each vehicle. In all the plots that follow, we will show three curves.  One that corresponds to the mean performance of the algorithm with the 90 percentile error bars. The second curve is the CRLB evaluated given all the measurements. The third bound is the ``causal'' CRLB computed by considering for parameters at each time instant $t$ (i.e. $\{u_t(k)\}_{k=1}^N$) only the measurements until time $t$. The y-axis in all the plots is the localization error (m). \\
        \begin{figure}
 \centering
 \includegraphics[height = 2.5in]{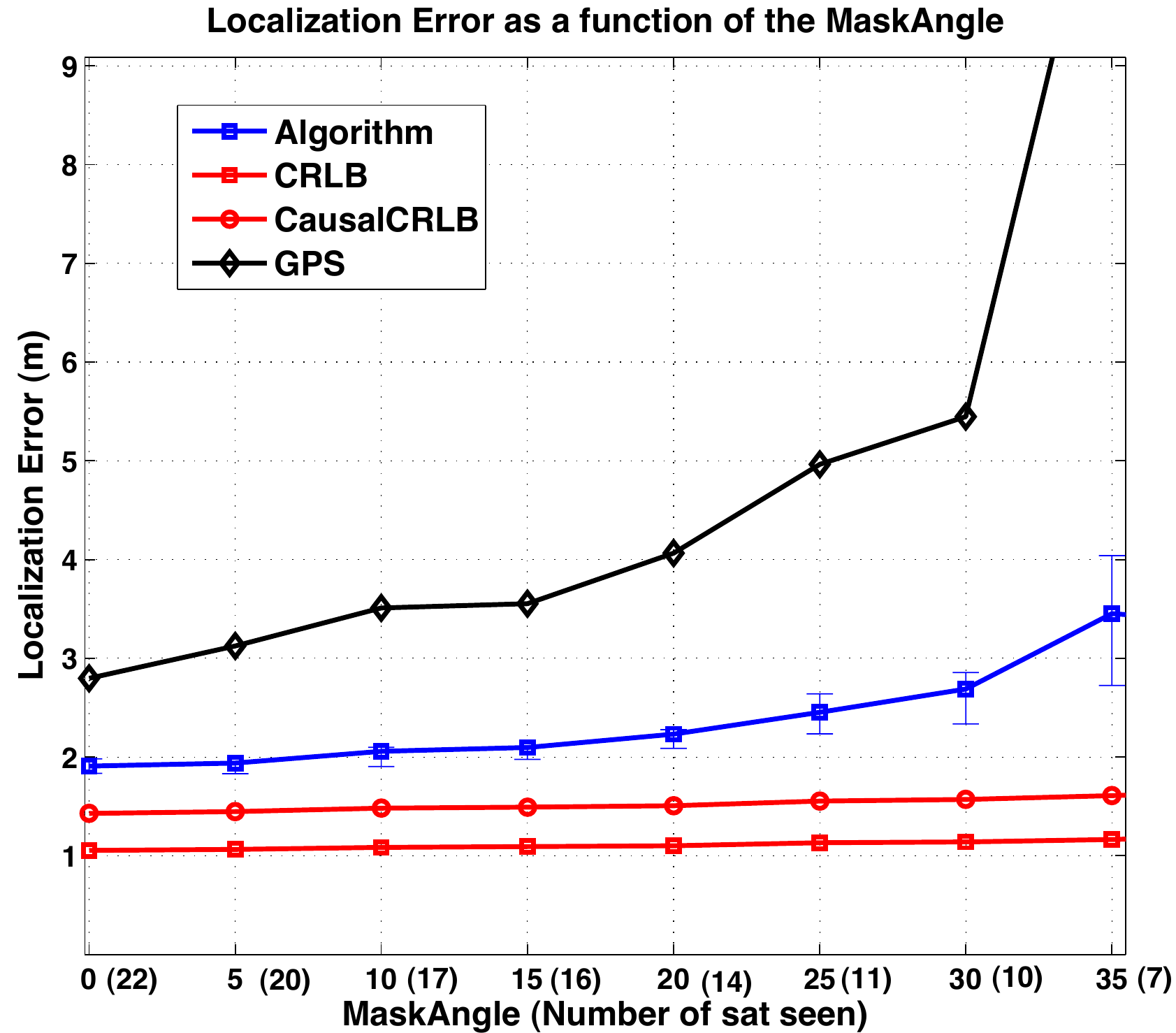}
 \caption{Performance of the algorithm as a function of the mask angle (number of satellites seen). $\sigma_{sat} = 10m, $$\sigma_{R} = 2m$, $\sigma_{INS} = 5m/s$, $\alpha = 0.5$, radius of communication $= 10m \sim 5$ neighbors.}
 \label{fig:ErrorMaskAngle}
 \vspace*{0in} \end{figure}
 
%       \begin{figure}
% \centering
% \includegraphics[height = 2in]{CoopPosNLOS_therm_10_alpha_05_ranging_3_therm}
% \caption{Performance of the algorithm as a function of the satellite thermal noise. $\sigma_{R} = 3m$, $\sigma_{INS} = 1m/s$, $\alpha = 0.5$. Mask angle is taken to be 30deg.}
% \label{fig:ErrorTherm}
% \vspace*{0in} \end{figure}
 
       \begin{figure}
 \centering
 \includegraphics[height = 2.5in]{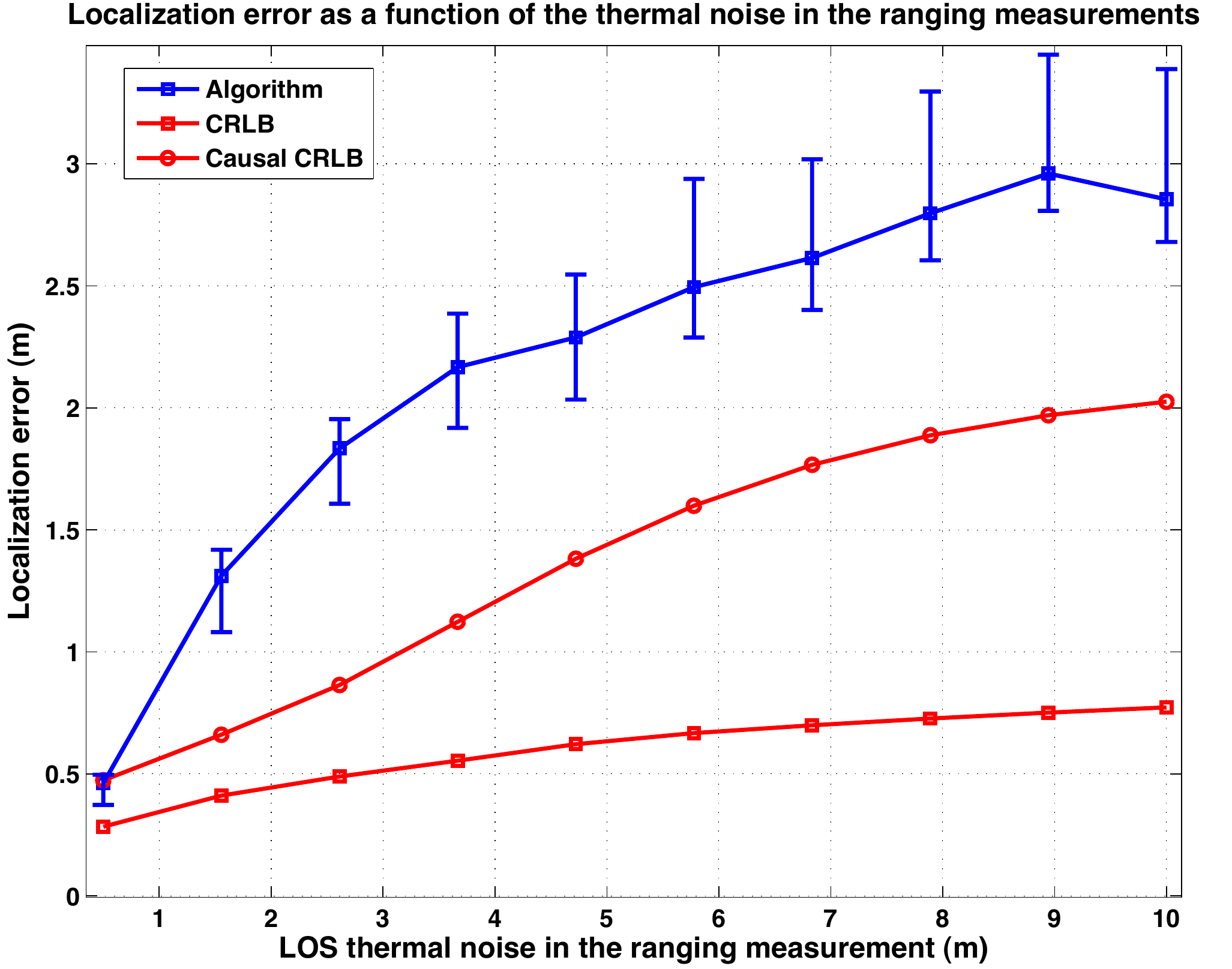}
 \caption{Performance of the algorithm as a function of the ranging noise. $\sigma_{sat} = 10m$, $\sigma_{INS} = 1m/s$, $\alpha = 0.5$.}
 \label{fig:ErrorRanging}
 \vspace*{0in} \end{figure}
 
 Fig. \ref{fig:ErrorMaskAngle} shows the error performance as a function of the mask angle (equivalently number of visible satellites) as compared against a standard least squares GPS receiver algorithm which is significantly worse off at lower satellite visibility mainly due to the NLOS nature of the measurements. In this simulation, we also simulated the Galileo satellite constellation to have a larger dynamic range of visible satellites. Fig. \ref{fig:ErrorRanging} shows the error performance as a function of the thermal noise in the ranging measurement between vehicles. At very low values of $\sigma_R$, the performance of the algorithm is close to the lower bound. However as $\sigma_R$ increases, the performance quickly deviates from the bound though it follows the scaling trend without diverging.\\
    \begin{figure}
 \centering
 \includegraphics[height = 2.5in]{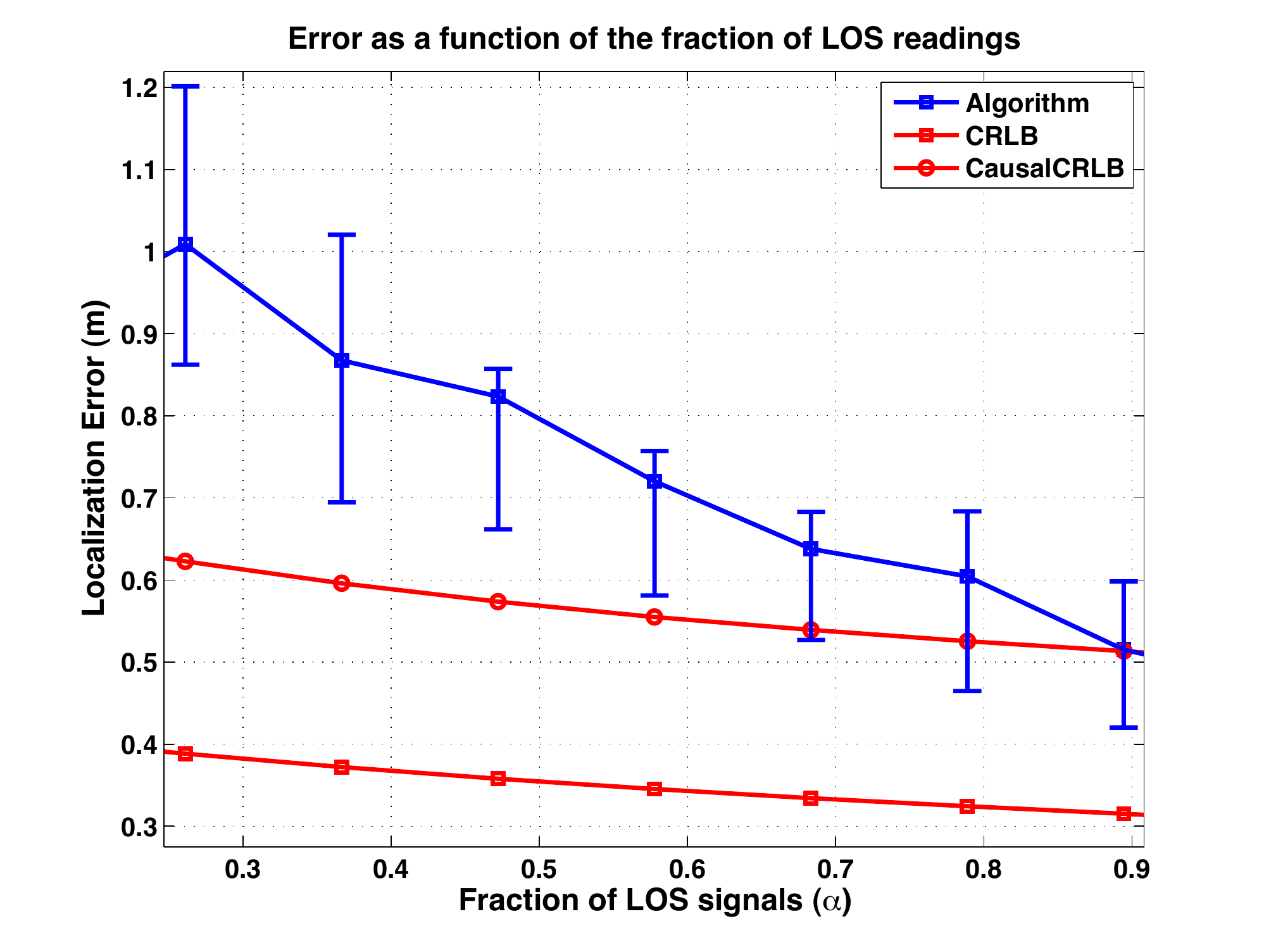}
 \caption{Performance of the algorithm as a function of the fraction of LOS measurements. $\sigma_{sat} = 10m$, $\sigma_R = 1m$, $\sigma_{INS} = 1m/s$.}
 \label{fig:ErrorAlpha}
 \vspace*{0in} \end{figure}
 
    \begin{figure}
 \centering
 \includegraphics[height = 2.5in]{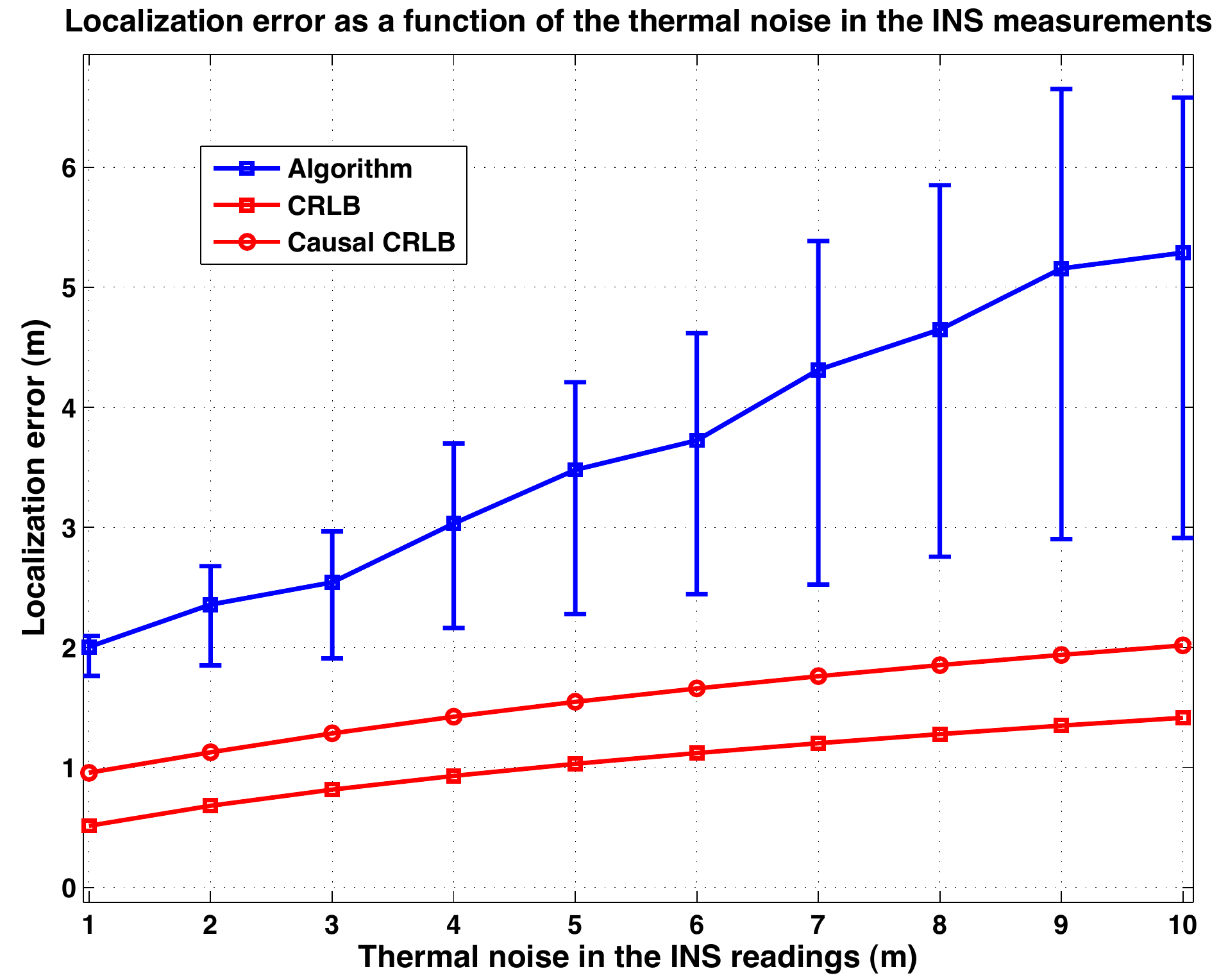}
 \caption{Performance of the algorithm as a function of thermal noise in the INS measurements. $\sigma_{sat} = 10m$, $\sigma_R = 3m$, $\alpha = 0.5$.}
 \label{fig:ErrorINS}
 \vspace*{0in} \end{figure}
 
  Fig. \ref{fig:ErrorAlpha} is a plot of the error performance as a function of the fraction of LOS measurements. The errors are seen to be quite reasonable over a wide range of $\alpha$. As before, the performance is close to the bound at higher signal-to-noise ratio i.e. large $\alpha$. Fig. \ref{fig:ErrorINS} shows the performance of the algorithm as a function of the measurement noise in the INS readings. Clearly the performance diverges from the bound as the measurement noise becomes large. This is an inherent limitation of the algorithm given that each vehicle uses an estimate of the neighbor's location extrapolated using only the INS measurements. Given that the position has to sustain until enough LOS measurements are gathered, a larger INS error cannot be tolerated by the algorithm. The hope is that the INS measurement error is small enough in practical systems for the algorithm to work. Vehicle INS sensors can have error $< 1m/s$ \cite{rezaei2007kalman} which can be tolerated by the algorithm.\\
 
   \begin{figure}
 \centering
 \includegraphics[height = 2.5in]{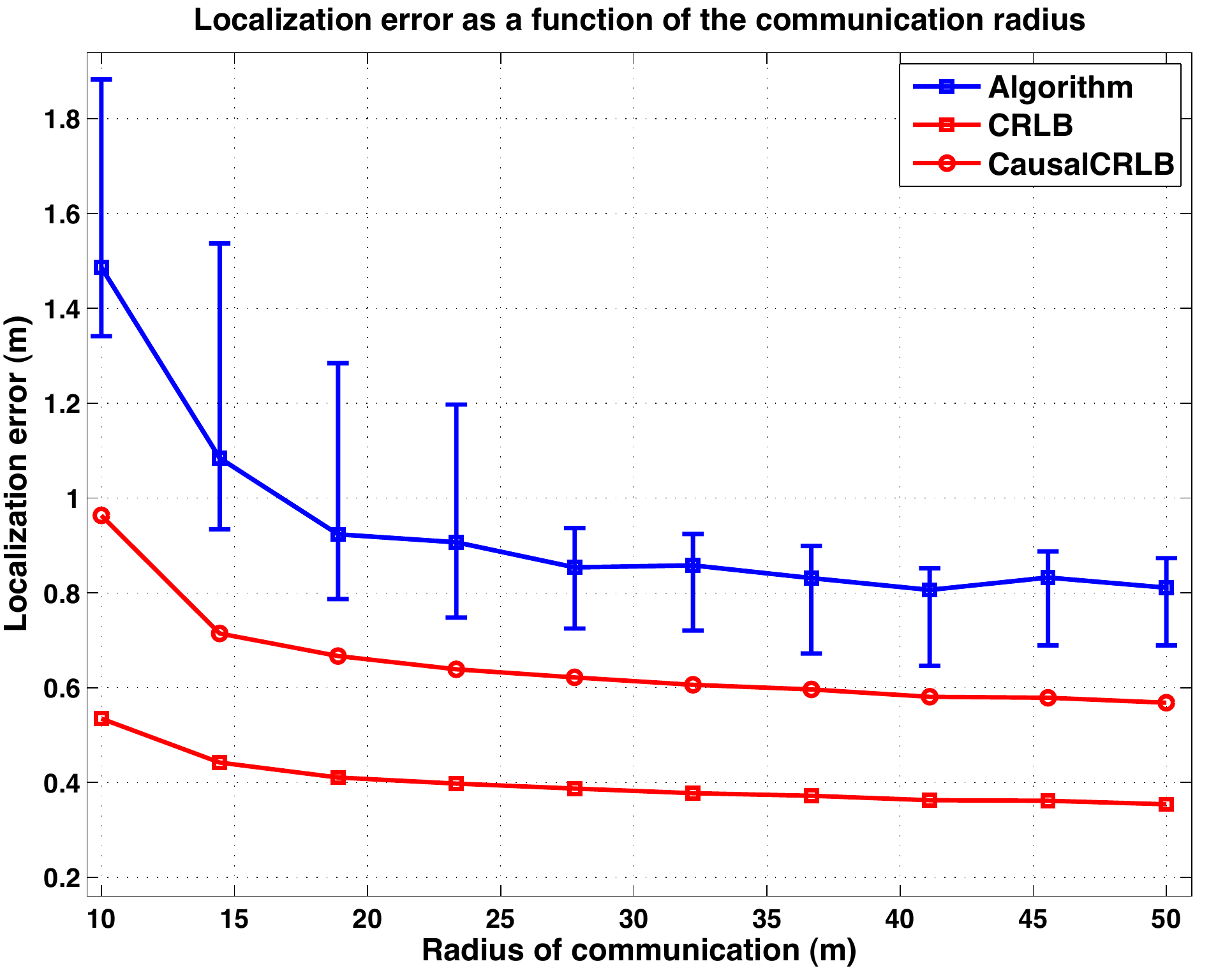}
 \caption{Performance of the algorithm as a function of the radius of communication. $\sigma_{sat} = 10m$, $\sigma_R = 1m$, $\sigma_{INS} = 1m/s, \alpha = 0.5$.}
 \label{fig:ErrorVeh}
 \vspace*{0in} \end{figure}
 
 Fig. \ref{fig:ErrorVeh} is a plot of the performance of the algorithm as a function of the communication radius. Essentially we are interested in seeing the performance improvement as a function of the increase in the number of neighbors of a node. As expected there is diminishing returns with increased radius. Having a large communication radius is good for localization. However this leads to decreased spatial reuse of the spectrum, thereby leading to a tradeoff. One could think of an adaptive system wherein the vehicles increase the communication radius in harsh environments like urban canyons while reduce it in open sky environments where the number of visible satellites is larger.\\
 
\begin{table}[htbp]
   \centering
   \begin{tabular}{cc} % Column formatting, @{} suppresses leading/trailing space
      \multicolumn{2}{c}{} \\
      \hline
      Ranging thermal Noise (m)    & Average $N_{eff}$\\
      \hline
      1      &  84.3 \\
        2        &   162.2 \\
      3       &  245.1 \\
      4       &  306.3 \\
      5 &   402.2 \\
      \hline
   \end{tabular}
   \caption{Effective number of particles ($N_{eff}$) as a function of the ranging noise}
   \label{tab:Neff}
\end{table}
 
 \begin{figure}
 \centering
 \includegraphics[height = 3in]{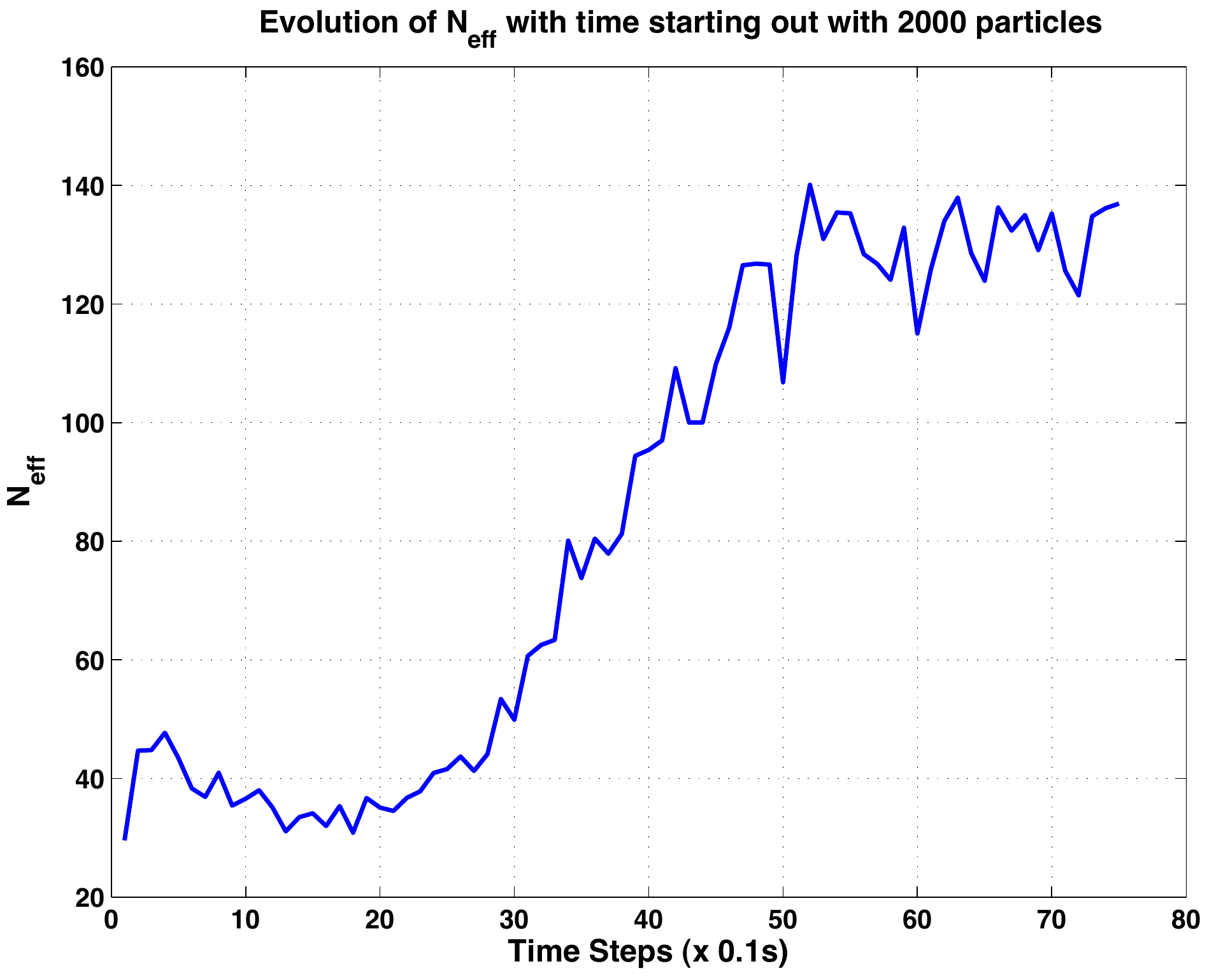}
 \caption{Effective number of particles ($N_{eff}$) as a function of time.}
 \label{fig:Neff}
 \vspace*{0in} \end{figure} 

 Particle filtering inherently has the problem of degeneracy where the number of distinct particles reduce with time. Thus it is important to understand the evolution of distinct particles for the problem setup.  We used 2000 particles to start with and the effective number of particles is around 100 depending on the parameters of the simulation. We set a threshold of 30 before we resample. Table \ref{tab:Neff} and Fig. \ref{fig:Neff}show the variation of $N_{eff}$ as a function of different parameters. For a larger noise in general, $N_{eff}$ is larger as expected since more number of particles will have a significant weight given the larger uncertainty in the position. However clearly the error is also larger. A smaller noise gives better error performance but the degeneracy problem kicks in earlier which would later adversely affect the error performance. $N_{eff}$ in  Fig. \ref{fig:Neff} increases over time since, after resampling we would have a larger number of distinct particles closer to the mean value and thereby each of them will have a significant weight.\\
 
 \begin{figure}[h!]
\centering
\includegraphics[height=3in]{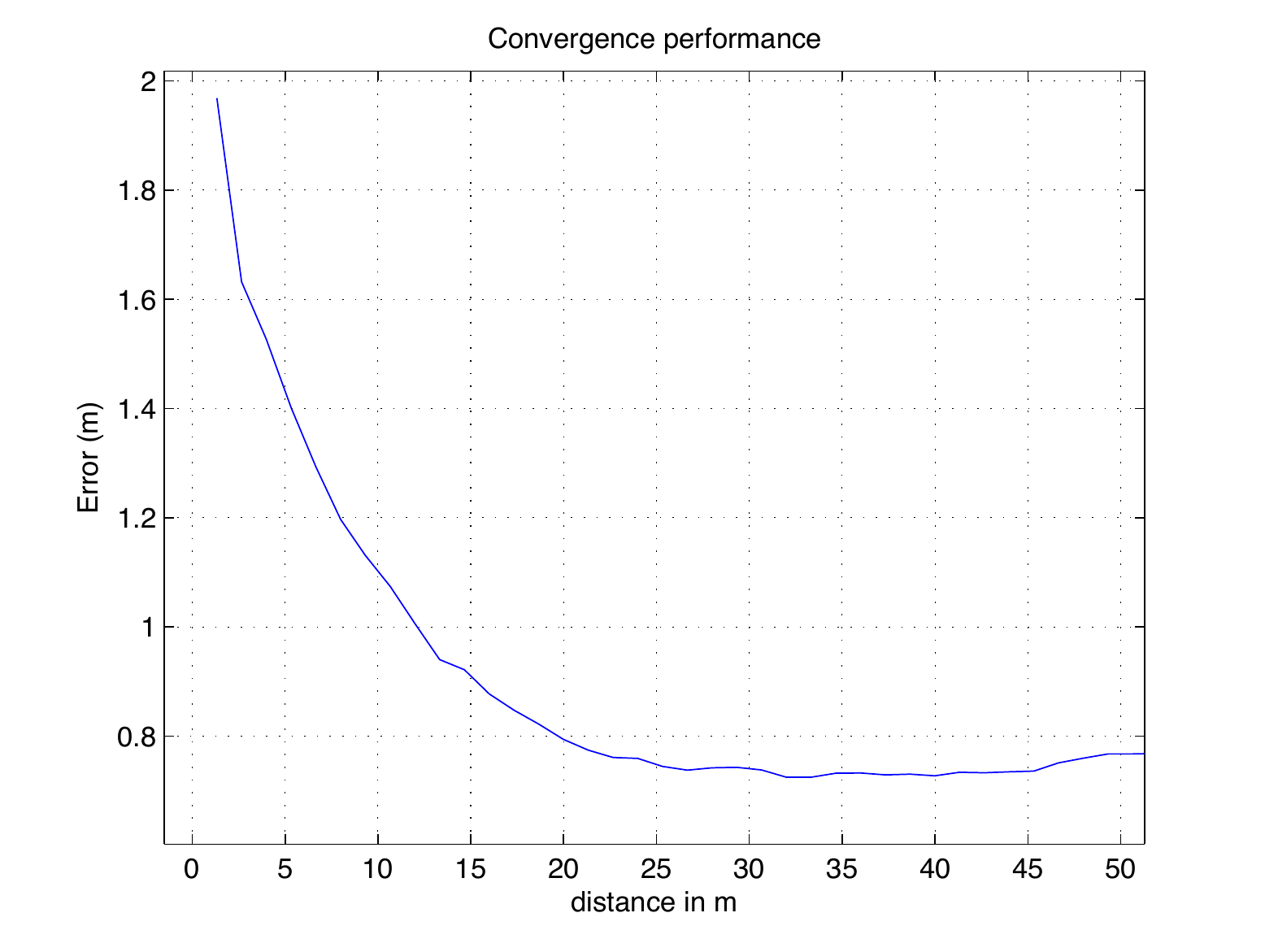}
\caption{Error convergence as a function of distance.}
\label{fig:conv}
\end{figure}
\begin{figure}[h!]
\centering
\includegraphics[height=3in]{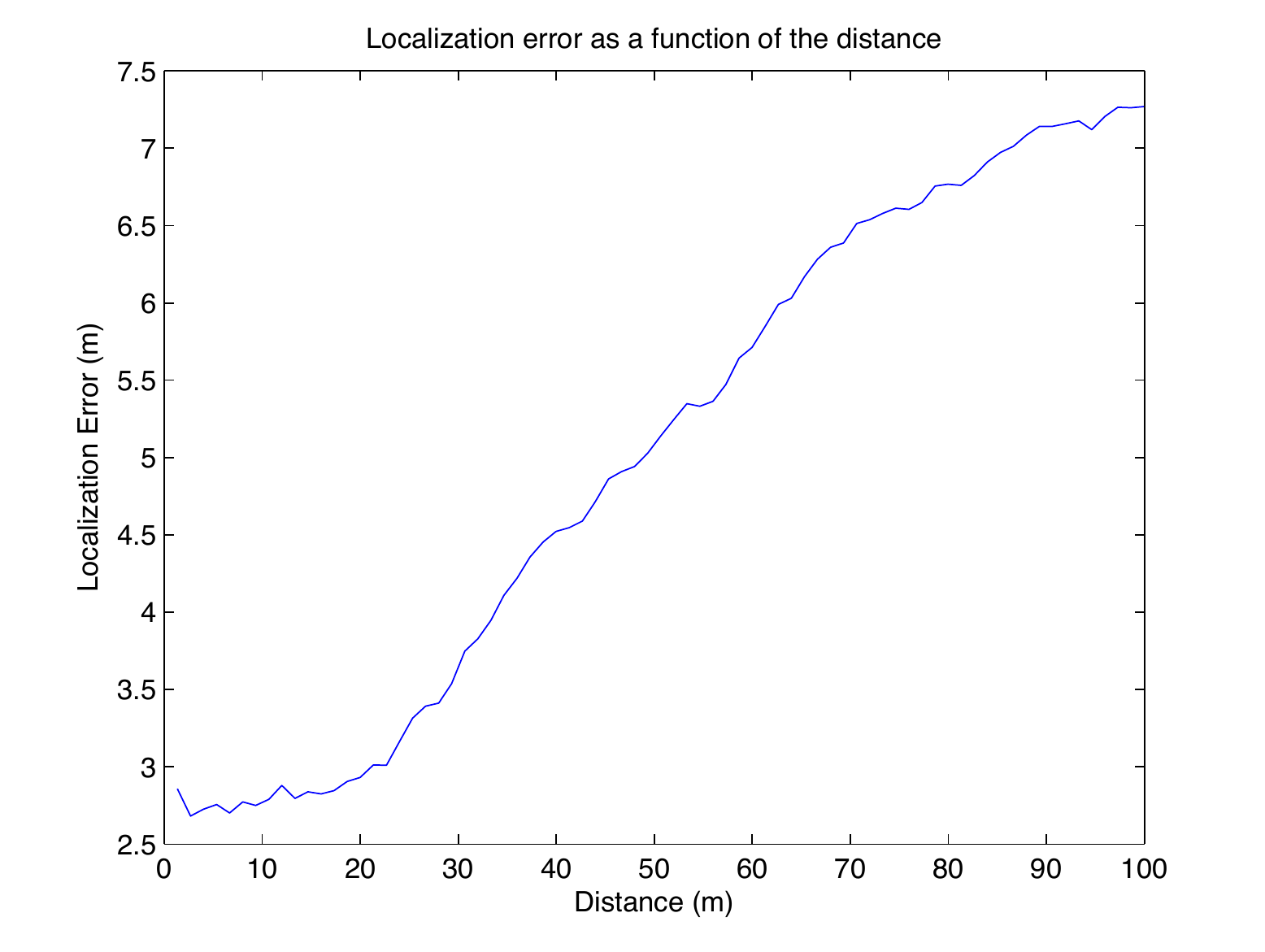}
\caption{Algorithm diverges when the noise is large.}
\label{fig:div}
\end{figure}

 Given the approximate nature of the algorithm, it is likely that the algorithm can diverge for certain parameters of the problem setup. Hence it is important to see how well the errors converge and when would the performance diverge. Fig. \ref{fig:conv} shows the error convergence as a function of the distance (equivalently time steps) for $\sigma_R = 1m$, $\sigma_{sat} = 10m$, $\sigma_{INS} = 1m/s$, $\alpha = 0.5$ and radius of communication $ = 50m$.  Fig. \ref{fig:conv}  shows the error as a function of the distance for the case when $\sigma_R = 3m$, $\sigma_{INS} = 10m/s$. Clearly the algorithm diverges for a larger noise in the INS readings. This is due to the inherent nature of the algorithm, since each vehicle extrapolates the neighbor's location using the INS measurements and uses that to compute its own belief. Thus a larger INS noise can lead to divergence.

 \section{Conclusion and Future work}
In this work we explored the application of particle filtering to get estimates of vehicle locations in a highly NLOS environment. We derived weight update equations for the NLOS setting and simulation results show that reasonably good accuracies in positioning is feasible. The approximation in the graphical model could break down above a certain noise threshold and below a certain anchor density, and the algorithm could potentially diverge. A theoretical understanding of when the algorithm diverges is another research direction, though we believe this to be a hard problem. We also explored the behavior of the localization error as a function of the number of anchors, vehicles and the fraction of LOS measurements by analyzing the CRLB. The performance of the proposed algorithm was evaluated in comparison to the derived bound and was shown that the algorithm respects the scaling behavior as predicted by the bound. One can see that the errors converge in around 10 steps.

%\section*{Acknowledgment}
%
%The authors would like to thank the anonymous reviewers and the TPC reviewers for their detailed and useful comments that helped improve the quality of the paper.

\bibliographystyle{IEEEtran}
\bibliography{JournalCoopLoc}
%\onecolumn
 \begin{appendices}

  \section{Derivation of the CRLB for Anchors}
Recall that ${\bf u}$ is the vector of agent/vehicle locations. We will treat the anchors separately and will denote the vector of anchor locations by ${\bf v}$.  Let $ \bf{x} = \left[ \begin{array}{c}
\bf{u} \\
\bf{v} \\
 \end{array} \right] \in {\mathbb C}^{N+M}$ be the vector of all the node locations. The location difference, $x_i-x_j$, between any two nodes can be described using the $N \times M$ vector ${\bf e}_{ij}$, whose $i$th entry is 1 and $j$th entry is -1 and all other entries are set to zero. Thus we get, 
 \bea {\bf e}_{ij}^{T}{\bf x} = x_i - x_j.\eea 
Let $L$ be the total number of distance measurements obtained in the network. We will assume that distance measurements are obtained between nodes that are within a communication radius $R$ of each other. Collecting all the location differences we get the following relation,
\bea E\bf{x} = \bf{y},\eea  where the rows of $E \in {\mathbb R^{L \times (N+M)}}$, are the vectors ${\bf e}_{ij}^{T}$ and $\bf{y} \in  {\mathbb C^{L \times 1}} $, is a complex vector of all the available location differences of nodes that are within a radius $R$ of each other. The absolute value of each entry in the vector ${\bf y}$ denotes the distance between the corresponding two nodes obtained from the matrix $E$.\\

 Let $E = [E_1 E_2]$, where $E_1 \in {\mathbb R}^{L \times N}$ and $E_2 \in {\mathbb R}^{L \times M}$. We can then write,
\bea E_1 {\bf u} + E_2 {\bf v} = {\bf y}.\eea   Let ${\bf d}$ denote the vector of all distances between nodes having observations i.e. ${\bf d = |y|}$, where ${\bf |y|}$ is a notation used to denote a vector whose components are the absolute values of the individual components of ${\bf y}$. Let ${\bf \hat{d}}$ denote the vector of  pairwise distance measurements  i.e.
\bea \hat{d}_j = d_j + n_j,\eea  where $n_j \sim p_{NOISE}(.)$.
 Define the two real diagonal matrices,
\bea 
D_R &\triangleq& Re\{\mbox{diag}\{y_1/|y_1|.......y_L/|y_L|\}\}\\
D_I &\triangleq& Im\{\mbox{diag}\{y_1/|y_1|.......y_L/|y_L|\}\}.
\eea 
 The authors in \cite{weiss2008maximum} obtain a compact representation for the Fisher Information Matrix for LOS Gaussian noise as shown below,
\bea  F = \frac{1}{\sigma^2}\left[ \begin{array}{c c}
E_{1}^TD_{R}^2E_1 &  E_{1}^TD_{R}D_{I}E_1 \\
E_{1}^TD_{R}D_{I}E_1 &  E_{1}^TD_{I}^2E_1 \\
 \end{array} \right],\eea 
 $F$ is assumed to be invertible. In our generalized case we get,
 \bea  F = \frac{1}{g(p_{NOISE})}\left[ \begin{array}{c c}
E_{1}^TD_{R}^2E_1 &  E_{1}^TD_{R}D_{I}E_1 \\
E_{1}^TD_{R}D_{I}E_1 &  E_{1}^TD_{I}^2E_1 \\
 \end{array} \right],\eea 
 
 Since we are interested in analyzing the behavior as a function of the number of anchors, we will ignore the scalar  $\frac{1}{g(p_{NOISE})}$ and work only with the matrix which we denote as $F_G$. Let us suppose that we add a single anchor to the set of existing nodes. Let $l$ be the number of additional distance measurements that are obtained. Without loss of generality assume that the first $l$ nodes get measurements with respect to the new anchor. The equation relating the agent locations to the location differences, $E{\bf x} = {\bf y}$ now gets updated to,
 \bea \left[ \begin{array}{c c c}
E_{1} &  E_{2}  &   0\\
\mbox{I}_l | \bf{0} &  \bf{0} &  -\bf{1}  \\
 \end{array} \right]   \left[ \begin{array}{c}
\bf{x} \\
x_{N+M+1} \\
 \end{array} \right]  =   \left[ \begin{array}{c}
\bf{y} \\
y_{L+1} \\
..\\
y_{L+l}
 \end{array} \right] . \eea 
 where $\mbox{I}_l$ is the identity matrix of dimension $l$,  ${\bf 0}$ and ${\bf 1}$ are vectors/matrices consisting of all zeros and all ones respectively, of appropriate dimensions. Let $\Delta_1 \triangleq [{\bf I}_l | {\bf 0}]$. Define
 \bea 
D_{R_1} &\triangleq& Re\{\mbox{diag}\{y_{L+1}/|y_{L+1}|.......y_{L+l}/|y_{L+l}|\}\}\\
D_{I_1} &\triangleq& Im\{\mbox{diag}\{y_{L+1}/|y_{L+1}|.......y_{L+l}/|y_{L+l}|\}\}
\eea 

Thus, after adding a single anchor, the matrix $E_1$ gets updated to $\left[\begin{array}{c} E_1\\ \mbox{I}_l | \bf{0} \end{array}\right]$. Hence the entries of the Fisher Matrix get updated as follows. We have $E_1^T D_R^2 E_1$ getting updated to $\left[\begin{array}{c} E_1\\ \mbox{I}_l | \bf{0} \end{array}\right]^T \left[\begin{array}{c c} D_R^2 & 0\\ 0 & D_{R_1}^2 \end{array}\right] \left[\begin{array}{c} E_1\\ \mbox{I}_l | \bf{0} \end{array}\right]$. Defining $\Delta_1 =  [\mbox{I}_l | \bf{0}]$, the new Fisher matrix can be written as,
\bea \tilde{F}_G = \left[ F_G +  \left[ \begin{array}{c c}
\Delta_{1}^T{D}_{R_1}^2\Delta_1 &  \Delta_{1}^T{D}_{R_1}{D}_{I_1}\Delta_1 \\
\Delta_{1}^T{D}_{R_1}{D}_{I_1}\Delta_1&  \Delta_{1}^T{D}_{I_1}^2\Delta_1 \\
 \end{array} \right]\right].\eea 

Note that in this case we defined a particular ordering of the nodes, i.e. the first $l$ nodes to have measurements with the newly added anchor. For a general ordering, it is easy to see that the identity matrix in the definition of $\Delta_1$ would be replaced by a general permutation matrix. Thus, if we add $M'$ anchors recursively in the network, the new Fisher matrix can be expressed as,
\bea \tilde{F}_G =  F_G +\sum_{i=1}^{M'}  \left[ \begin{array}{c c}
\Delta_{i}^T{D}_{R_i}^2\Delta_i &  \Delta_{i}^T{D}_{R_i}{D}_{I_i}\Delta_i \\
\Delta_{i}^T{D}_{R_i}{D}_{I_i}\Delta_i &  \Delta_{i}^T{D}_{I_i}^2\Delta_i \\
 \end{array} \right].\eea 
 Here the matrices $\Delta_{i}$ of size $N \times l_i$ where $l_i$ nodes get measurements with the $i$th newly added anchor, have ones corresponding to the columns of the nodes with which the $i$th newly introduced anchor gets measurements. ${D}_{R_i}, {D}_{I_i}$ have definition similar to ${D_{R_1}}$ and ${D_{I_1}}$. For simplicity lets first consider the case where the newly introduced anchors have measurements with all the agents. We then have $\Delta_{i} = I_N \ \ \forall i$ giving us,
 \bea \tilde{F}_G =  F_G +\sum_{i=1}^{M'}  \left[ \begin{array}{c c}
{D}_{R_i}^2 &  {D}_{R_i}{D}_{I_i} \\
{D}_{R_i}\tilde{D}_{I_i} &  {D}_{I_i}^2 \\
 \end{array} \right].\eea 
By definition, $D_{R_i}(j) = \frac{Re(y_k)}{|y_k|}$ where $k = L + (i-1)N+j$. This can also be equivalently written as $D_{R_i}(j) = \mbox{cos}(\phi_{ij})$ where $\phi_{ij}$ is the angle made by the line joining $i$th newly added anchor node and the $j$th agent, with the horizontal axis. Let us assume that each anchor that is newly introduced is randomly placed in the field independent of all other nodes. Thus for each node $j$, $\{\phi_{ij}\}_{i=1}^{M'}$'s are i.i.d and distributed U$(0,2\pi)$.  Similarly   $D_{I_i}(j) = \mbox{sin}(\phi_{ij})$ and by the strong law of large numbers we get,
 \bea 
 \sum_i D_{R_i}^2(j) & = & \sum_{i = 1}^{M'} \cos^2(\phi_{ij})  \rightarrow  \frac{M'}{2},\\
  \sum_i D_{I_i}^2(j) & = & \sum_{i = 1}^{M'} \sin^2(\phi_{ij})  \rightarrow  \frac{M'}{2},\\
    \sum_i D_{I_i}(j) D_{R_i}(j) & = & \sum_{i = 1}^{M'} \sin(\phi_{ij})\cos(\phi_{ij})  \rightarrow  0.\\
 \eea 
 
 The Fisher matrix now simplifies to,
 \bea \tilde{F}_G = F_G + \frac{M'}{2} I _{2N} \ \ w.h.p.\eea 
  We  know that the Fisher Information Matrix is a covariance matrix and hence is symmetric positive definite. We can write $F_G = U \Lambda U^H$, where $\Lambda$ is a diagonal matrix of the eigen values of $F_G$ and $UU^H = U^HU = I$. We also know that $Tr(ABC) = Tr(CAB) = Tr(BCA)$.
  This gives us, 
  \bea 
  \mbox{Trace}(\tilde{F_G}^{-1}) &=& \mbox{Trace}(U \Lambda U^H +  \frac{M'}{2} I _{2N})^{-1}, \\
                                                    &=& \mbox{Trace}(U (\Lambda +  \frac{M'}{2} I _{2N})U^H)^{-1},\\
                                                    & = & \sum_{i=1}^{2N}  \frac{1}{\lambda_i + \frac{M'}{2}}.
  \eea 
  It is now easy to extend the analysis to the case where each anchor node has measurements only with the nodes that are within a radius $R$. Let $A$ be the total area of the field where the nodes are placed. Let  $\rho = \frac{\pi R^2}{A}$ and then $\rho M'$ would be the average number of neighbors of each node. In this case we would have $\Delta_{i}^T{D}_{R_i}^2\Delta_i \rightarrow \frac{\rho M'}{2} I_N$ and so on. One can easily show that,
  \bea  \mbox{Trace}(\tilde{F_G}^{-1}) = \sum_{i=1}^{2N}  \frac{1}{\lambda_i + \frac{\rho M'}{2}}.\eea 
  
  \section{Derivation of the CRLB for agents}
  The setup is similar as in the previous case with $N$ agents, $M$ anchors and $L$ distance measurements. We are interesting in characterizing the behavior of the CRLB after adding $N'$ agents to the existing network of anchors and agents. Consider the simple case when $N' = 1$. We will assume for now that the newly introduced agent has measurements with all of the existing $N$ agents and $M$ anchors. Let $\tilde{F_G}$, $\tilde{E}$, $\tilde{E_1}$, $\tilde{D_R}$, $\tilde{D_I}$, be the new set of matrices obtained after adding this node. We then have the following,
\bea  \tilde{E} = 
\left[
\begin{array}{ccc}
  E_1 & {\bf 0}_L  & E_2   \\
  -I_N& {\bf 1}_N  & 0   \\
  {\bf 0}_N & {\bf 1}_M & -I_M\\
\end{array}
\right]
,\eea 
\bea \tilde{E_1} = 
\left[
\begin{array}{cc}
 E_1 & {\bf 0}_L \\
  -I_N & {\bf 1}_N  \\
   {\bf 0}_N & {\bf 1}_M\\
\end{array}
\right].
\eea 
Let 
\bea 
D_{R_{11}} &\triangleq& Re\{\mbox{diag}\{y_{L+1}/|y_{L+1}|.......y_{L+N}/|y_{L+N}|\}\}\\
D_{R_{12}} &\triangleq& Re\{\mbox{diag}\{y_{L+N+1}/|y_{L+N+1}|.......y_{L+N+M}/|y_{L+N+M}|\}\}\\
D_{I_{11}} &\triangleq& Im\{\mbox{diag}\{y_{L+1}/|y_{L+1}|.......y_{L+N}/|y_{L+N}|\}\}\\
D_{I_{12}} &\triangleq& Im\{\mbox{diag}\{y_{L+N+1}/|y_{L+N+1}|.......y_{L+N+M}/|y_{L+N+M}|\}\}
\eea 
Then\bea  \tilde{D}_R = 
\left[
\begin{array}{ccc}
 D_R & {\bf 0}  & {\bf 0}\\
 {\bf 0} & D_{R_{11}} & {\bf 0} \\
 {\bf 0} &  {\bf 0} & D_{R_{12}}\\
\end{array}
\right]
,\eea 
\bea  \tilde{D}_I = 
\left[
\begin{array}{ccc}
 D_I & {\bf 0}  & {\bf 0}\\
 {\bf 0} & D_{I_{11}} & {\bf 0} \\
 {\bf 0} &  {\bf 0} & D_{I_{12}}\\
\end{array}
\right].
\eea 

The new Fisher Information Matrix is given by,
\bea  \tilde{F_G} = \left[ \begin{array}{c c}
\tilde{E}_{1}^T\tilde{D}_{R}^2\tilde{E}_1 &  \tilde{E}_{1}^T\tilde{D}_{R}\tilde{D}_{I}\tilde{E}_1 \\
\tilde{E}_{1}^T\tilde{D}_{R}\tilde{D}_{I}\tilde{E}_1 &  \tilde{E}_{1}^T\tilde{D}_{I}^2\tilde{E}_1 \\
 \end{array} \right].\eea 

The individual terms of $\tilde{F_G}$ can be simplified as shown in (\ref{eqn:A}) - (\ref{eqn:B}) (lengthy equations are in the last page). \\

Lets now consider adding one more agent to the existing set of agents i.e. $N' = 2$. The second agent gets measurements from the first $N$ agents and $M$ anchors. In this case we get the following updated matrices,
\bea \tilde{E_1} &=& 
\left[
\begin{array}{ccc}
 E_1 & {\bf 0}_L &   {\bf 0}_L \\
  -I_N & {\bf 1}_N &  {\bf 0}_N   \\
  {\bf 0}_N & {\bf 1}_M & {\bf 0}_M\\
   -I_N & {\bf 0}_N &  {\bf 1}_N\\
   {\bf 0}_N & {\bf 0}_M & {\bf 1}_M  
\end{array}
\right], \\
 \tilde{D}_R & = &  
\left[
\begin{array}{ccccc}
 D_R& {\bf 0}  & {\bf 0} &  {\bf 0} &  {\bf 0}\\
 {\bf 0} & D_{R_{11}} & {\bf 0} &  {\bf 0} &  {\bf 0} \\
  {\bf 0} & {\bf 0} & D_{R_{12}} & {\bf 0} &  {\bf 0} \\
    {\bf 0} & {\bf 0} & {\bf 0} &  D_{R_{21}} &   {\bf 0}\\
     {\bf 0} & {\bf 0} & {\bf 0} &    {\bf 0} & D_{R_{22}} 
 \end{array}
\right].
\eea 
The terms in the new Fisher Information Matrix can be simplified and have the structure shown in (\ref{eqn:B}). The Fisher matrix evolution as more and more nodes are added is apparent from the expression (\ref{eqn:B}). Similar evolution holds for other block terms in the Fisher matrix. To simplify the analysis it would be good if we could separate out the original Fisher Information Matrix terms and express $\tilde{F_G}$ in terms of $F_G$. This requires rearranging some of the terms in $\tilde{F_G}$. Recall the definition of $F_G$. We had ${\bf \eta = [u_{R}^T u_{I}^T]^T}$, where ${\bf u_R} = Re\{{\bf u}\}, {\bf u_I} = Im\{{\bf u}\}$. Then $F_G$ is given by,
\bea F_{G_{ij}} \triangleq E\left\{  \frac{\partial f({\bf \hat{d} | \eta})}{\partial \eta_i}   \frac{\partial f({\bf \hat{d} | \eta})}{\partial \eta_j} \right\}.\eea 
Let ${\bf z = z_R + j z_I} \in {\mathbb C^{N' \times 1}}$, denote the location of the newly added nodes. The Fisher Information Matrix, $\tilde{F_G}$ that we have calculated corresponds to the following ordering of the parameters, 
\bea {\bf \tilde{\eta} = \left[
\left[
\begin{array}{c}
 u_R\\
 z_R
\end{array}
\right]^T 
\left[
\begin{array}{c}
 u_I\\
 z_I
\end{array}
\right]^T \right]^T}.\eea 

We will now rearrange the parameters so as to get, 
\bea {\bf {\tilde\tilde\eta}} = \left[\begin{array}{c} u_{R} \\ 
u_{I}\\
 z_{R}\\ 
 z_{I}
 \end{array}\right].\eea 
Retaining the same notation for $\tilde{F_G}$, we get the following simplification, 
\bea \tilde{F_G} = 
\left[
\begin{array}{cc}
  F_G + \Delta_{11}& \Delta_{12}   \\
  \Delta_{21}& \Delta_{22}   \\      
\end{array}
\right],
\eea 
where, 
\bea 
\Delta_{11} =  
\left[
\begin{array}{cc}
  \sum_{j=1}^{N'} D_{R_{j1}}^2 &  \sum_{j=1}^{N'} D_{R_{j1}} D_{I_{j1}}     \\
 \sum_{j=1}^{N'} D_{R_{j1}} D_{I_{j1}}  &    \sum_{j=1}^{N'} D_{I_{j1}}^2 \\   
\end{array}
\right].
\eea 
$\Delta_{22}$ and $\Delta_{12}$ are given by the expressions (\ref{eqn:C}) and (\ref{eqn:D}) respectively.\\

We are now interested in the error improvement of the first $N$ agents after the addition of $N'$ agents. For this, it is sufficient to look at the Schur Complement of the matrix $\Delta_{22}$, since the inverse of the Schur Complement corresponds to the CRLB restricted to the first $N$ nodes. The Schur Complement is given by,
 \bea F + \Delta_{11}-\Delta_{12}\Delta_{22}^{-1}\Delta_{21}.\eea  
  Based on similar arguments as in the case of anchor nodes, assuming that each of the newly added nodes are distributed uniform i.i.d into the network, we have for large $N'$,
\bea \Delta_{11} \rightarrow \frac{N'}{2} I_{2N}.\eea 
Let us now assume that the initial set of nodes were also placed uniform randomly in the field. Consider the terms in $\Delta_{22}$. Each of the terms $ {\bf 1}_{N}^T D_{R_{j1}}^2 {\bf 1}_N$ are the sum of the cosine of the angles made by the newly introduced $j$th node with all the existing nodes in the network. Under the random placement assumption, these angles can also be taken to be distributed i.i.d $U(0,2\pi)$. Thus we have for large $N$,
\bea \Delta_{22} \rightarrow \frac{N+M}{2}I_{2N'}.\eea 
Note the difference in this approach as compared to that for the anchor nodes. Here we are averaging over all initial node placements also for this approximation to hold. For the anchor nodes, the result was true for any initial node placement.\\

With the above approximation in place, we have
\bea \Delta_{12}\Delta_{22}^{-1}\Delta_{12}^T = \frac{2}{N+M}\Delta_{12}\Delta_{12}^T.\eea  
This can be expanded to obtain the expression (\ref{eqn:E}). With the usual law of large numbers argument, each of the terms in the matrix converge as $N'$ grows large. The corresponding values to which the  $kl$th term in the matrix converges are shown in (\ref{eqn:F})-(\ref{eqn:G}).
Thus we get,
\bea \Delta_{12}\Delta_{12}^T \rightarrow \frac{N'}{4}
\left(
\left[
\begin{array}{cc}
  {\bf 1}_N{\bf 1}_{N}^T&  {\bf 0}  \\
 {\bf 0} &    {\bf 1}_N{\bf 1}_{N}^T \\
\end{array}
\right] + I_{2N}\right).
\eea 
The Schur complement can now be simplified as shown in (\ref{eqn:H}).
The CRLB restricted to the first $N$ nodes is given by the inverse of this Schur Complement. Now consider the case when measurements are obtained only between nodes that are within a radius $R$ of each other. Let $\rho = \frac{\pi R^2}{A}$, then similar arguments would simplify the new Fisher matrix $F'$ to be,
 \bea F + \frac{\rho N'}{2}\left( \left(1-\frac{1}{\rho(N+M)}\right)I_{2N} - \frac{1}{\rho(N+M)}\left[
\begin{array}{cc}
  {\bf 1}_N{\bf 1}_{N}^T&  {\bf 0}  \\
 {\bf 0} &    {\bf 1}_N{\bf 1}_{N}^T \\
\end{array}
\right]\right).\eea 

\section{Mobile setting}
Let ${\bf u}_{R_t}+ j{\bf u}_{I_t} $ be the vector of vehicle locations at time $t$. Define the parameter vector as ${\bf \eta} = \left[
\begin{array}{rl} {\bf u}_{R_1};{\bf u}_{I_1} ; ....; {\bf u}_{R_T} ; {\bf u}_{I_T} \end{array}\right]$. Let  ${\bf s}$ be the vector of velocity measurements. Since the measurement noise is independent across the velocity measurements and the ranging measurements, it is easy to see that the $(k,m)$th entry in the Fisher Information Matrix is given by,
\bea
F_{km} & = & \mathbb{E}\left\{  \frac{\partial \ln p({\Theta | \eta})}{\partial \eta_k}   \frac{\partial \ln p({\Theta | \eta})}{\partial \eta_m} \right\} + \mathbb{E}\left\{  \frac{\partial \ln p({\ {\bf s} | \eta})}{\partial \eta_k}   \frac{\partial \ln p({\ {\bf s}| \eta})}{\partial \eta_m} \right\},\\
\eea
Define the matrix $F_G$ such that,
\bea
F_{G_{km}} & = &  \mathbb{E}\left\{  \frac{\partial \ln p({\Theta | \eta})}{\partial \eta_k}   \frac{\partial \ln p({\Theta | \eta})}{\partial \eta_m} \right\}, 
\eea 
and matrix $F_{INS}$ such that,
\bea
F_{INS_{km}} & = & \mathbb{E}\left\{  \frac{\partial \ln p({\ {\bf s} | \eta})}{\partial \eta_k}   \frac{\partial \ln p({\ {\bf s}| \eta})}{\partial \eta_m} \right\}.
\eea
$F_G$ is the Fisher Information Matrix corresponding only to the ranging measurements and $F_{INS}$ is the Fisher Information Matrix characterizing the contribution of the inertial navigation system measurements to positioning. Consider the matrix $F_G$. Since the measurement noise is independent across vehicles as well as across time, one can verify that, for $(\eta_k = u_{R_t}, \eta_m = u_{R_{t'}})$, we have that $F_{G_{km}} = 0$ whenever $t \neq t'$. Similarly this holds true for other pairs of parameters $(\eta_k,\eta_m)$ that corresponds to different time instants. Thus, we would end up with a block diagonal structure of the following form,
 \bea F_G =    \left[ \begin{array}{cccc} % brackets may be (...), [...], \{...\}, or left out
       F_G(1) & 0 &...& 0\\
       0 & F_{G}(2) & ... & 0 \\
       0 & 0 & ... & 0\\
       0 & 0 & ...& F_{G}(T)\\
    \end{array}\right] ,\eea
where $F_G(t)$ is the Fisher Information Matrix as derived in Appendix A, with the vehicle locations at time $t$. 
$F_{INS}$ can be evaluated as follows.
    Due to the independence of noise, the only non-zero terms in the matrix would be the diagonal and the first off-diagonal entries. One can verify that for $\eta_k = u_{R_t}(k)$, $\eta_m = u_{R_t}(k)$, $F_{INS_{km}} = \frac{2}{\sigma_{INS}^2}$, for $\eta_k = u_{R_t}(k)$, $\eta_m = u_{R_{t-1}}(k)$,  $F_{INS_{km}} = \frac{-1}{\sigma_{INS}^2}$ and for $\eta_k = u_{R_t}(k)$, $\eta_m = u_{R_{t+1}}(k)$,  $F_{INS_{km}} = \frac{-1}{\sigma_{INS}^2}$. Similarly this holds for the imaginary components.

\newpage
%\begin{figure*}[h!]
%\vspace*{-0.25in}
{\scriptsize
\bea
\tilde{E}_{1}^T\tilde{D}_{R}^2\tilde{E}_1  & = &\left[
\begin{array}{ccc}
 E_1^T & -I_N &  {\bf 0}_N^T \\
{\bf 0}_L^T &  {\bf 1}_N^T &  {\bf 1}_M^T\\
\end{array}
\right]
\label{eqn:A}
\left[
\begin{array}{ccc}
 D_R & {\bf 0}  & {\bf 0}\\
 {\bf 0} & D_{R_{11}} & {\bf 0} \\
 {\bf 0} &  {\bf 0} & D_{R_{12}}\\
\end{array}
\right]
\left[
\begin{array}{cc}
 E_1 & {\bf 0}_L \\
  -I_N & {\bf 1}_N  \\
   {\bf 0}_N & {\bf 1}_M\\
\end{array}
\right] \\
& = & 
\left[
\begin{array}{cc}
E_{1}^TD_{R}^2 E_1 + D_{R_{11}}^2  &   -D_{R_{11}}^2 {\bf 1}_N   \\
 -{\bf 1}_{N}^T D_{R_{11}}^2 &   {\bf 1}_{N}^T D_{R_{11}}^2 {\bf 1}_N+{\bf 1}_{M}^T D_{R_{12}}^2 {\bf 1}_M
\end{array}
\right]\\
\tilde{E}_{1}^T\tilde{D}_{I}^2\tilde{E}_1 & = & 
\left[
\begin{array}{cc}
E_{1}^TD_{I}^2 E_1 + D_{I_{11}}^2  &   -D_{I_{11}}^2 {\bf 1}_N   \\
 -{\bf 1}_{N}^T D_{I_{11}}^2 &   {\bf 1}_{N}^T D_{I_{11}}^2 {\bf 1}_N +  {\bf 1}_{M}^T D_{I_{12}}^2 {\bf 1}_M
\end{array}
\right]
\\
\tilde{E}_{1}^T\tilde{D}_R\tilde{D}_{I}\tilde{E}_1 & = &
\left[
\begin{array}{cc}
E_{1}^TD_{R}D_{I} E_1 + D_{R_{11}} D_{I_{11}}&   -D_{R_{11}}D_{I_{11}} {\bf 1}_N   \\
 -{\bf 1}_{N}^T D_{R_{11}}D_{I_{11}} &   {\bf 1}_{N}^T D_{R_{11}}D_{I_{11}} {\bf 1}_N +  {\bf 1}_{M}^T D_{R_{11}}D_{I_{11}} {\bf 1}_M
\end{array}
\right]
\eea
}

{\scriptsize
\bea
\tilde{E}_{1}^T\tilde{D}_{R}^2\tilde{E}_1  & = & \left[
\begin{array}{ccccc}
 E_1^T &  -I_N &  {\bf 0}_{N}^T  & -I_N & {\bf 0}_{N}^T\\
  {\bf 0}_{L}^T & {\bf 1}_{N}^T  & {\bf 1}_{M}^T &  {\bf 0}_{N}^T   &  {\bf 0}_{M}^T\\
  {\bf 0}_{L}^T  &  {\bf 0}_{N}^T  &  {\bf 0}_{M}^T & {\bf 1}_{N}^T & {\bf 1}_{M}^T 
\end{array}
\right]
\left[
\begin{array}{ccccc}
 D_R& {\bf 0}  & {\bf 0} &  {\bf 0} &  {\bf 0}\\
 {\bf 0} & D_{R_{11}} & {\bf 0} &  {\bf 0} &  {\bf 0} \\
  {\bf 0} & {\bf 0} & D_{R_{12}} & {\bf 0} &  {\bf 0} \\
    {\bf 0} & {\bf 0} & {\bf 0} &  D_{R_{21}} &   {\bf 0}\\
     {\bf 0} & {\bf 0} & {\bf 0} &    {\bf 0} & D_{R_{22}} 
 \end{array}
\right]
\left[
\begin{array}{ccc}
 E_1 & {\bf 0}_L &   {\bf 0}_L \\
  -I_N & {\bf 1}_N &  {\bf 0}_N   \\
  {\bf 0}_N & {\bf 1}_M & {\bf 0}_M\\
   -I_N & {\bf 0}_N &  {\bf 1}_N\\
   {\bf 0}_N & {\bf 0}_M & {\bf 1}_M  
\end{array}
\right] \\
& = & 
\left[
\begin{array}{ccc}
E_{1}^TD_{R}^2 E_1 + D_{R_{11}}^2 + D_{R_{21}}^2  &   -D_{R_{11}}^2 {\bf 1}_N   & -D_{R_{21}}^2 {\bf 1}_N \\
 -{\bf 1}_{N}^T D_{R_{11}}^2 &   {\bf 1}_{N}^T D_{R_{11}}^2 {\bf 1}_N + {\bf 1}_{M}^T D_{R_{12}}^2 {\bf 1}_M &  {\bf 0}\\
  -{\bf 1}_{N}^T D_{R_{21}}^2 &  {\bf 0} &   {\bf 1}_{N}^T D_{R_{21}}^2 {\bf 1}_N +  {\bf 1}_{M}^T D_{R_{22}}^2 {\bf 1}_M \\
\end{array}
\right]
\label{eqn:B}
\eea}

{\scriptsize
\bea
\Delta_{22} & = &   
\left[
\begin{array}{ccc|ccc}
  {\bf 1}_{N}^T D_{R_{11}}^2 {\bf 1}_N &  {\bf 0} &    {\bf 0}        &   {\bf 1}_{N}^T D_{R_{11}} D_{I_{11}} {\bf 1}_N   & {\bf 0}  &  {\bf 0}  \\
   {\bf 0} & {\bf 1}_{N}^T D_{R_{21}}^2 {\bf 1}_N & {\bf 0} &   {\bf 0} &{\bf 1}_{N}^T D_{R_{21}} D_{I_{21}} {\bf 1}_N  &  {\bf 0}   \\
              &  ....... & & &....&  \\
    {\bf 0}&{\bf 0}   & {\bf 1}_{N}^T D_{R_{N'1}}^2 {\bf 1}_N  &   {\bf 0} &  {\bf 0}  &{\bf 1}_{N}^T D_{R_{N'1}} D_{I_{N'1}} {\bf 1}_N    \\
    \hline
    
  {\bf 1}_{N}^T D_{R_{11}}D_{I_{11}} {\bf 1}_N &  {\bf 0} & {\bf 0}           &   {\bf 1}_{N}^T  D_{I_{11}}^2 {\bf 1}_N   & {\bf 0}  &  {\bf 0}  \\
   {\bf 0} & {\bf 1}_{N}^T D_{R_{21}}D_{I_{21}} {\bf 1}_N & {\bf 0} &   {\bf 0} &{\bf 1}_{N}^T D_{I_{21}}^2 {\bf 1}_N  &  {\bf 0}   \\
              &  ....... & &  &....& \\
    {\bf 0}& {\bf 0}  & {\bf 1}_{N}^T D_{R_{N'1}}D_{I_{N'1}} {\bf 1}_N  &   {\bf 0} &  {\bf 0}  &{\bf 1}_{N}^T D_{I_{N'1}}^2 {\bf 1}_N    \\
\end{array}
\right] \\
& + &    
\left[
\begin{array}{ccc|ccc}
  {\bf 1}_{M}^T D_{R_{12}}^2 {\bf 1}_M &  {\bf 0} &    {\bf 0}        &   {\bf 1}_{M}^T D_{R_{12}} D_{I_{12}} {\bf 1}_M   & {\bf 0}  &  {\bf 0}  \\
   {\bf 0} & {\bf 1}_{M}^T D_{R_{22}}^2 {\bf 1}_M & {\bf 0} &   {\bf 0} &{\bf 1}_{M}^T D_{R_{22}} D_{I_{22}} {\bf 1}_M  &  {\bf 0}   \\
              &  ....... & & &....&  \\
    {\bf 0}&{\bf 0}   & {\bf 1}_{M}^T D_{R_{N'2}}^2 {\bf 1}_M  &   {\bf 0} &  {\bf 0}  &{\bf 1}_{M}^T D_{R_{N'2}} D_{I_{N'2}} {\bf 1}_M    \\
    \hline
    
  {\bf 1}_{M}^T D_{R_{12}}D_{I_{12}} {\bf 1}_M &  {\bf 0} & {\bf 0}           &   {\bf 1}_{M}^T  D_{I_{12}}^2 {\bf 1}_M   & {\bf 0}  &  {\bf 0}  \\
   {\bf 0} & {\bf 1}_{M}^T D_{R_{22}}D_{I_{22}} {\bf 1}_M & {\bf 0} &   {\bf 0} &{\bf 1}_{M}^T D_{I_{22}}^2 {\bf 1}_M  &  {\bf 0}   \\
              &  ....... & &  &....& \\
    {\bf 0}& {\bf 0}  & {\bf 1}_{M}^T D_{R_{N'2}}D_{I_{N'2}} {\bf 1}_M  &   {\bf 0} &  {\bf 0}  &{\bf 1}_{M}^T D_{I_{N'2}}^2 {\bf 1}_M    \\
\end{array}
\right]
\label{eqn:C}
\eea
}

{\scriptsize
\bea
%\small
\Delta_{12} & = & 
\left[
\begin{array}{cccc|cccc}
    D_{R_{11}}^2 {\bf 1}_N&  D_{R_{21}}^2 {\bf 1}_N&... &D_{R_{N'1}}^2 {\bf 1}_N&D_{R_{11}}D_{I_{11}} {\bf 1}_N&  D_{R_{21}}D_{I_{21}} {\bf 1}_N&... &D_{R_{N'1}}D_{I_{N'1}} {\bf 1}_N  \\
 D_{R_{11}}D_{I_{11}} {\bf 1}_N&  D_{R_{21}}D_{I_{21}} {\bf 1}_N&... &D_{R_{N'1}}D_{I_{N'1}} {\bf 1}_N&  D_{I_{11}}^2 {\bf 1}_N&  D_{I_{21}}^2 {\bf 1}_N&... &D_{I_{N'1}}^2 {\bf 1}_N       \\        
\end{array}
\right]
\label{eqn:D}
\eea
}
{\scriptsize
\bea
\Delta_{12}\Delta_{12}^T & = &  
\left[
\begin{array}{c|c}
  \sum_{j=1}^{N'} D_{R_{j1}}^2 {\bf 1}_N {\bf 1}_{N}^T D_{R_{j1}}^2 +  D_{R_{j1}}D_{I_{j1}} {\bf 1}_N {\bf 1}_{N}^T D_{R_{j1}} D_{I_{j1}}&   \sum_{j=1}^{N'} D_{R_{j1}}^2 {\bf 1}_N {\bf 1}_{N}^T D_{R_{j1}}D_{I_{j1}} +  D_{R_{j1}}D_{I_{j1}} {\bf 1}_N {\bf 1}_{N}^T D_{I_{j1}}^2\\
 \sum_{j=1}^{N'} D_{R_{j1}}^2 {\bf 1}_N {\bf 1}_{N}^T D_{R_{j1}}D_{I_{j1}} +  D_{R_{j1}}D_{I_{j1}} {\bf 1}_N {\bf 1}_{N}^T D_{I_{j1}}^2 &   \sum_{j=1}^{N'} D_{I_{j1}}^2 {\bf 1}_N {\bf 1}_{N}^T D_{I_{j1}}^2 +  D_{R_{j1}}D_{I_{j1}} {\bf 1}_N {\bf 1}_{N}^T D_{R_{j1}} D_{I_{j1}}    \\
\end{array}
\right]
\label{eqn:E}
\eea
}

{\scriptsize
 \bea
 \left(\sum_{j=1}^{N'} D_{R_{j1}}^2 {\bf 1}_N {\bf 1}_{N}^T D_{R_{j1}}^2 \right)(kl) & = &  \sum_{j=1}^{N'} \cos(\phi_j(k))^2 \cos(\phi_j(l))^2 \\
  \label{eqn:F}
                            & \rightarrow & 
\left\{
\begin{array}{cc}
 \frac{3N'}{8} &   \mbox{if  } k=l \\
 \frac{N'}{4} & \mbox{o.w}    \\    
\end{array}
\right.\\
 \left(\sum_{j=1}^{N'} D_{R_{j1}}^2 {\bf 1}_N {\bf 1}_{N}^T D_{R_{j1}}D_{I_{j1}} + D_{R_{j1}}D_{I_j} {\bf 1}_N {\bf 1}_{N}^T D_{I_{j1}}^2 \right)(kl) & = & \sum_{j=1}^{N'} ( \cos(\phi_j(k))^2 \cos(\phi_j(l))\sin(\phi_j(l))+\cos(\phi_j(k)) \sin(\phi_j(k))\sin(\phi_j(l))^2)\\
   &\rightarrow& 0\\
  \left(\sum_{j=1}^{N'} D_{R_{j1}}D_{I_{j1}} {\bf 1}_N {\bf 1}_{N}^T D_{R_{j1}} D_{I_{j1}}\right)(kl)    & = &  \frac{1}{4}\sum_{j=1}^{N'} \sin\phi_j(l)) \sin\phi_j(k))  \\
                            & \rightarrow &  
\left\{
\begin{array}{cc}
 \frac{N'}{8} & \mbox{if  } j =k     \\
  0 &   \mbox{o.w} \\   
\end{array}
\right. \\
   \left(\sum_{j=1}^{N'} D_{R_{j1}}D_{I_{j1}} {\bf 1}_N {\bf 1}_{N}^T D_{I_{j1}}^2 \right)(kl) 
   &\rightarrow& 0
   \label{eqn:G}
  \eea}

{\scriptsize
\bea
F + \Delta_{11}-\Delta_{12}\Delta_{22}^{-1}\Delta_{21} \rightarrow F + \frac{N'}{2}\left( \left(1-\frac{1}{N+M}\right)I_{2N} - \frac{1}{N+M}\left[
\begin{array}{cc}
  {\bf 1}_N{\bf 1}_{N}^T&  {\bf 0}  \\
 {\bf 0} &    {\bf 1}_N{\bf 1}_{N}^T \\
\end{array}
\right]\right)
\label{eqn:H}
\eea
}
%\end{figure*}

\end{appendices}

\end{document}

%% file: packages_all.tex
\usepackage{graphicx}
\usepackage{color}
\usepackage{chapterbib}
\usepackage{pdfsync}
\usepackage[hang]{subfigure}

\usepackage{times}
\usepackage{latexsym}
\usepackage{bbm,dsfont}
\usepackage{amsmath}
\usepackage{amssymb}
\usepackage{latexsym}
\usepackage{amsthm}
\usepackage{amsfonts}
\usepackage{epsfig}
\usepackage{latexsym}
\usepackage{cite}
\usepackage{graphicx}
\usepackage{amsmath}
\usepackage{amsthm}

\newcommand{\bea}{\begin{eqnarray}}
\newcommand{\eea}{\end{eqnarray}}

\usepackage{epsf}
\usepackage{amssymb}
\usepackage{graphicx}
\usepackage{amsmath}
\usepackage[english]{babel}
\usepackage{mathrsfs}
\newtheorem{thm}{Theorem}

\usepackage{setspace}

\input{newcommands}

\usepackage{ifthen}
\usepackage{color}
\usepackage{amssymb}

\usepackage[newenum]{paralist}

\newcommand{\IsTR}{\ifthenelse{1<2}}

\IsTR{}{ \ninept }

%% file: newcommands.tex
\newcommand{\bfl}{\begin{flushleft}}
\newcommand{\efl}{\end{flushleft}}
\newcommand{\bfr}{\begin{flushright}}
\newcommand{\efr}{\end{flushright}}
\newcommand{\bc}{\begin{center}}
\newcommand{\ec}{\end{center}}

\newcommand{\edar}{\end{array}}
\newcommand{\begar}{\begin{array}}

\newcommand{\bit}{\begin{itemize}}
\newcommand{\eit}{\end{itemize}}
\newcommand{\beq}{\begin{equation}}
\newcommand{\eeq}{\end{equation}}
\newcommand{\ben}{\begin{enumerate}}
\newcommand{\een}{\end{enumerate}}
\newcommand{\bean}{\begin{eqnarray*}}
\newcommand{\eean}{\end{eqnarray*}}

\newtheorem{cor}{Corollary}